\documentclass[a4paper,12pt]{article}
\usepackage{geometry}

\usepackage[utf8]{inputenc}
\usepackage{amsmath}
\usepackage{amssymb}
\usepackage{bbm}
\usepackage{graphicx}
\usepackage{accents}
\usepackage{amsthm}
\usepackage{color}
\usepackage{graphics}
\usepackage{epsfig}
\usepackage{cancel}
\usepackage{epstopdf}
\usepackage{hyperref}
\usepackage{xcolor}
\usepackage{mathrsfs}
\usepackage{stmaryrd}

\usepackage{upgreek}
\usepackage{mathbbol}

\numberwithin{equation}{section}

\usepackage{xcolor}
\usepackage{framed}
\usepackage{shadethm}
\definecolor{shadecolor}{gray}{.9}
\definecolor{shaderulecolor}{gray}{.1}

\DeclareMathOperator{\tr}{tr}
\DeclareMathOperator{\dev}{dev}
\DeclareMathOperator{\pd}{\partial}
\newcommand{\mbf}[1]{\mathbf{#1}}
\newcommand{\mfk}[1]{\mathfrak{#1}}
\newcommand{\mcl}[1]{\mathcal{#1}}

\newtheorem{remark}{Remark}

\newtheorem{theorem}{Theorem}
\newtheorem{claim}{Claim}

\usepackage{hyperref}
\usepackage{xcolor}
\hypersetup{
colorlinks,
linkcolor={blue},
citecolor={blue},
urlcolor={blue!80!black}
}

\newcommand{\mcal}[1]{\mathcal{#1}}
\newcommand{\mbb}[1]{\mathbb{#1}}

\newcommand{\bsym}[1]{\boldsymbol{#1}}

\DeclareMathOperator{\up}{\bsym{u}}
\DeclareMathOperator{\us}{\bsym{v}}

\usepackage[left]{lineno}
\setpagewiselinenumbers

\usepackage[symbol]{footmisc}
\makeatletter
\def\blfootnote{\xdef\@thefnmark{}\@footnotetext}
\makeatother

\usepackage{array}
\newcolumntype{L}[1]{>{\raggedright\let\newline\\\arraybackslash\hspace{0pt}}m{#1}}
\newcolumntype{C}[1]{>{\centering\let\newline\\\arraybackslash\hspace{0pt}}m{#1}}
\newcolumntype{R}[1]{>{\raggedleft\let\newline\\\arraybackslash\hspace{0pt}}m{#1}}

\usepackage{tabu}

\title{Null Lagrangians in linear theories of micropolar type and few other generalizations of elasticity
}

\usepackage{fancyhdr}
\makeatletter
\let\runauthor\@author
\let\runtitle\@title
\makeatother
\author{Nirupam Basak$^1$, Basant Lal Sharma$^2$\thanks{Corresponding author email: bls@iitk.ac.in}\\
{\small $^1$Department of Mathematics, Indian Institute of Technology Kanpur}\\
{\small Kanpur, U. P. 208016, India; basakn@iitk.ac.in}\\
{\small $^2$Department of Mechanical Engineering, Indian Institute of Technology Kanpur}\\
{\small Kanpur, U. P. 208016, India; email{bls@iitk.ac.in}}}

\date{}

\allowdisplaybreaks 

\begin{document}

\maketitle



\begin{abstract}
{In the context of linear theories of generalized elasticity including those for homogeneous micropolar media, quasicrystals, piezoelectric and piezomagnetic media, we explore the concept of null Lagrangians. 
For obtaining the family of null Lagrangians we employ the sufficient conditions of H. Rund.
In some cases a non-zero null Lagrangian is found and
the stored energy admits a split into a null Lagrangian and a remainder. However, the null Lagrangian vanishes whenever the relevant elasticity tensor obeys certain symmetry conditions which can be construed as an analogue of the Cauchy relations. 
}
\end{abstract}

\section*{Introduction}
In the three-dimensional theory of classical linear elasticity, the function
\[
{\Psi}(\boldsymbol\epsilon)=\frac{1}{2}\boldsymbol\epsilon{\cdot} {{\mathbb{C}}}[\boldsymbol\epsilon]
\]
represents the stored energy, where ${{\mathbb{C}}}$ is the elasticity tensor \cite{Lo27} and the symmetric second order tensor $\boldsymbol\epsilon$ depends on displacement vector field. It is known \cite{Lancia95} that, for no choice of ${\mathbb{C}}$ other than $\mathbb{0}$, with the symmetries typical of elasticity tensors can it correspond to a null Lagrangian.
From a physical point of view, here, we recall that null Lagrangian refers to an energy density that does not contribute to the equilibrium equation for a given energy function \cite{Ericksen62} but rather accounts for environmental boundary effects. 
As an illustrative and significant example in the history of this subject, the topic of Frank constant $k_{24}$ in the classical theory of liquid crystals is well known one \cite{Virga94}; indeed, Frank constant (in the role of material modulus) has been shown to be associated with a term that is a null Lagrangian \cite{Ericksen62}.
As null Lagrangians cannot be interpreted as stored energy functionals in the linear theory of elasticity, i.e., borrowing a term from \cite{Ericksen62}, nilpotent energies are absent. 
However, as found out by \cite{cT94} (and \cite{Lancia95}) a given elasticity tensor ${\mathbb{C}}$ can be decomposed so that the stored energy functional admits a split into a null Lagrangian \cite{dE62,Ericksen62} and a remainder which further reveals interesting connections with the six Cauchy's relations \cite{cT94,Lo27,Lancia95,FI02,FI13}.
Exploring along these lines, we naturally attempt to widen the scope of these investigations to generalized theories of elasticity.
We recall that the classical theory of elasticity considers material particles as simple points having three translational degrees of freedom.
The constitutive relations of such materials are characterized by symmetric Cauchy stress as a consequence of the Cauchy hypothesis where surface loading is determined by force vector, while neglecting surface couples \cite{Lo27,Gurtin81}.
Occasionally, the classical model is found to be insufficient for some deformable solids such as cellular materials, foam-like structures (e.g., bones), periodic lattices, etc \cite{Lakes93b}. 
In this context, we invoke a three dimensional generalized elastic body which also has a rotational degree of freedom \cite{Cosserat09,ErT58,
Toupin64,Aero1,Aero2}.
In the case of linear micropolar material \cite{Eringen99}, the stored energy is represented by
\[
{{\Psi}}({{\upepsilon}},{{\upkappa}})=\frac{1}{2}{{\upepsilon}}{\cdot} {{\mathbb{A}}}[{{\upepsilon}}]+\frac{1}{2}{{\upkappa}}{\cdot} {{\mathbb{B}}}[{{\upkappa}}]
+{{\upepsilon}}{\cdot} {{\mathbb{D}}}[{{\upkappa}}],
\]
where ${{\mathbb{A}}},$ ${{\mathbb{B}}}$, ${{\mathbb{D}}}$ are the fourth order tensors of micropolar elasticity while ${{\upepsilon}}$ is the linearized strain tensor and ${{\upkappa}}$ is the linearized wryness tensor depending on the displacement vector field and the rotation vector field. 
In the first part of this document we discuss this energy functional from the point of view of null Lagrangians.
When ${{\mathbb{D}}}$ is non-zero it is commonly referred as a chiral (non-centrosymmetric) micropolar material
\cite{wN86}.
In addition to the linear micropolar elasticity, there have been also several developments towards generalization of elasticity for quasi-crystals \cite{BL75,ADL92}, piezoelectric and piezomagnetic media \cite{BCJ1964,ADL92}. We also discuss these theories from the viewpoint of nilpotent energies albeit briefly.
Such generalizations of the elasticity based models have important role in description of mechanical behaviour of the body in the presence of additional effects due to the presence of microstructure and other physical phenomenon. In fact, the micropolar kinematics itself is also deemed to be relevant for the piezoelectric or piezomagnetic materials since electromagnetic fields produce forces and moments \cite{EM90,Mg88}.
Based on the pioneering idea of \cite{Cosserat09}, several variants of the theory \cite{ErT58} of media with microstructure have been proposed \cite{Toupin64,Mindlin64,Eringen99}. Overall there have been a tremendous amount of researches within these classes of models, over the last century as well last couple of decades, which we do not intend to survey in this document.

We restrict our study to those Lagrangian functions ${\Psi}$ 
that belong to the null-space of the Euler--Lagrange operator; i.e., those functions for which the corresponding Euler--Lagrange equations are identically satisfied (see \cite{Rund66,Rund74,Olver88};
and also \cite{Ericksen62}, \cite{Picone62}, \cite{Picone64}, \cite{Franchis64}, \cite{Ball81}).
The characterization of the null Lagrangians goes back to the works of Caratheodory \cite{Caratheodory29} (so called `theory of equivalent integrals') and
given in \cite{Landers42,dE62}. The reader is also referred to the 
treatments given in \cite{Rund66} and \cite{dE80}.
It is also pertinent to mention the contribution of \cite{dE86},
who have given a representation for null Lagrangians
through the use of differential forms and suggested a way to map certain traction data to null Lagrangians.

\subsection*{Notation}
Let $\Omega$ denote an open, bounded region in three-dimensional space, $\bsym{n}$ be the outward unit normal to the boundary ${\pd}\Omega$. Also let ${{\bsym{u}}}$ denote the displacement vector field and ${\bsym{\phi}}$ be the rotation vector field over $\Omega$. 
Let $\{{\bsym{e}}_1, {\bsym{e}}_2, {\bsym{e}}_3\}$ be the standard basis for ${{\mathbb{R}}}^3$.
Let $\mathrm{SO(3)}$ denote the set of all rotation tensors in three dimensions, i.e., for $\mbf{Q}\in\mathrm{SO(3)}$ $\mbf{Q}^{\top}\mbf{Q}=\mbf{I}, \det \mbf{Q}=+1$. For a skew tensor $\mbf{W}$ (i.e., $\mbf{W}^{\top}=-\mbf{W}$) the axial vector
$\bsym{w}=\text{axl}\mbf{W},$
is defined by
$\mbf{W}\bsym{a}=\bsym{w}\wedge\bsym{a}, \forall \bsym{a}\in\mathbb{R}^3.$
The differentiation of a function $f$ with respect to the $x_j$ coordinate is written as $f_{,j}.$ Note that $\mbf{Q}^{\top}\mbf{Q}_{,j}$ is a skew tensor for a rotation tensor field $\mbf{Q}$ on $\Omega$.
A fourth order tensor, with components $A_{ijkl}$, is denoted by $\mathbb{A}$ and its action on the second order tensors is expressed as
$\mathbb{A}[\mbf{A}]$ with its $ij$ component given by $A_{ijkl}A_{kl}.$
The inner product between two second tensors $\mbf{A}$ and $\mbf{B}$ is denoted by $\mbf{A}\cdot\mbf{B}$ and equals $A_{ij}B_{ij}$.
The Einstein summation convention is used unless specified otherwise.
For convenience as well as additional clarity the notation for indices is changed at some places in the manuscript as specified later. Also, the notation and symbols for other relevant physical entities are defined as appropriate.

\section{Linearized micropolar model}
\label{micropolarelasticity}

In what follows we consider a {\em micropolar} elastic body where each point (particle) of the body is considered as {\em infinitesimal rigid body} with six degrees of freedom. The deformation of the body is described by a mapping from a fixed reference configuration into an current (deformed) configuration including rotations of the particles.
The reference configuration of such body is a bounded domain denoted by $\Omega\subset{\mathbb{R}}^3$.
We denote the {\em microdeformation} 
(describing the placement), 
and the {\em microrotation tensor} (describing the orientation of each particle) with
\begin{align}
\chi:\Omega\rightarrow\mathbb{R}^3, \quad
{\mbf{R}}:\Omega\rightarrow\mathrm{SO(3)},
\label{rotdef}
\end{align}
respectively.
It is assumed that $\Omega$ represents the reference configuration of the micropolar body. Consider any part $P$ with $\bsym{n}$ outward unit normal vector field on the surface $\pd P$ then
\begin{equation}
\bsym{t}={{\upsigma}}[\bsym{n}],\quad \bsym{m}={{\upmu}}[\bsym{n}],
\label{E:cauchy}
\end{equation}
where $\bsym{t}$ is the traction vector field, $\bsym{m}$ is the moment vector field, ${{\upsigma}}$ is the stress tensor field and ${{\upmu}}$ is the couple stress tensor field. In components,
$t_i={{\upsigma}}_{ij}n_j, m_i={{\upmu}}_{ij}n_j.$
For the micropolar body to be in equilibrium,
\begin{align}
\int_{P}\bsym{b}dc+\int_{{\pd} P}\bsym{t}ds=\bsym{0},
\text{ and }
\int_{P} ((\bsym{x}-\bsym{o})\wedge \bsym{b}+\bsym{c})dv+\int_{{\pd} P}((\bsym{x}-\bsym{o})\wedge\bsym{t}+\bsym{m})ds
=\bsym{0},
\end{align}
(with $\bsym{o}$ as the reference point) for every body part $P\subset\Omega$.
The equations of equilibrium in presence of an external body force (with components $b_i$) and an external body couple (with components $c_i$) are given by
\begin{align}
\label{microeqlb}
\frac{{\pd}{{\upsigma}}_{ij}}{{\pd} x_j} +b_i=0, \quad
\frac{{\pd} {{\upmu}}_{ij}}{{\pd}x_j}-{{\mathtt{E}}}_{ijk} {{\upsigma}}_{jk}+c_i=0.
\end{align}
Using \eqref{rotdef}, the deformation gradient and the nonsymmetric right stretch tensor are \cite{PE09} defined by
\begin{align}
\mbf{F}=\nabla\chi= \chi_{,i} \otimes {\bsym{e}}_i, \quad
\mbf{U}={\mbf{R}}^{\top}\mbf{F},
\end{align}
respectively.
In combination with the material frame indifference, this leads to the relative Lagrangian {\em stretch tensor}, as strain measure, 
and
the {\em wryness tensor}, as a Lagrangian measure for curvature,
respectively defined by
\begin{equation}
{\mbf{E}}={\mbf{U}}- {\mbf{I}},\quad
{\mbf{K}}= \text{axl}\big({\mbf{R}}^{\top} {\mbf{R}}_{,j} \big)\otimes {{\bsym{e}}}_j,
\end{equation}
In the nonlinear theory of micropolar elasticity, the strain energy density function \cite{PE09}
is 
${\mathcal{E}}({\mbf{E}},{\mbf{K}}).$
Let us consider the case when the translations ${{\bsym{u}}}$ and microrotations ${\bsym{\phi}}$ are very small, 
where 
\begin{align}
{{\bsym{u}}}:=\chi(\bsym{x})-\bsym{x},\quad
{\bsym{\phi}} := \varphi{\bsym{e}}
\end{align}
is the (small) translation vector and
the (small) rotation vector (i.e., $\mbf{R}\approx\mbf{I}-{{\mathtt{E}}}[{\bsym{\phi}}]$), respectively, so that
\begin{align}
{\mbf{E}}\approx\nabla{{\bsym{u}}}+{{\mathtt{E}}}[{\bsym{\phi}}], \quad{\mbf{K}}\approx\nabla{\bsym{\phi}},
\end{align}
where, ${\mathtt{E}}$ is the Levi-Civita (alternating) tensor in three dimensions (with components defined by ${{\mathtt{E}}}_{ijk}={\bsym{e}}_i\cdot{\bsym{e}}_j\wedge{\bsym{e}}_k$).
Above is essentially the definition of an asymmetric strain and an asymmetric torsion obtained by the linearization around the natural state \cite{Eringen99}. 
Indeed, the tensors ${{\upepsilon}}$, called the {\em linear stretch tensor} and ${{\upkappa}}$, called the {\em linear wryness tensor}, are defined by
\begin{equation}
{{\upepsilon}}=\nabla{{\bsym{u}}}+{{\mathtt{E}}}[{\bsym{\phi}}],\quad
\text{ and }
{{\upkappa}}=\nabla{\bsym{\phi}},
\label{E:akr}
\end{equation}
respectively. Thus, these are kinematic relations between the `small' displacement field ${{\bsym{u}}}$ and the `small' rotation vector ${\bsym{\phi}}$ vs the strain tensor (interpreted as linearized stretch tensor also) and torsion tensor (interpreted as linearized wryness tensor also);
in components ${{\upepsilon}}_{ij}={u}_{i,j}+{{\mathtt{E}}}_{kij}{\phi}_k$ and ${{\upkappa}}_{ij}={\phi}_{i,j}.$
We consider the case of 
homogeneous, linear, anisotropic, micropolar elasticity \cite{Eringen99}. The micropolar strain energy for a body of such a material is
\begin{align}
\label{en-grad1}
{{\mcl{V}}}({{{\bsym{u}}}},{{\bsym{\phi}}})=\int_\Omega {{\mathcal{E}}}({{\upepsilon}},{{\upkappa}}) dv,
\end{align}
with the energy density
\begin{align}
\label{en-grad2}
{{\mathcal{E}}}({{\upepsilon}},{{\upkappa}})=\frac{1}{2} {{A}}_{ijkl}{{\upepsilon}}_{ij} {{\upepsilon}}_{kl}
+\frac{1}{2} {{B}}_{ijkl} {{\upkappa}}_{ij}{{\upkappa}}_{kl}
+{{D}}_{ijkl} {{\upepsilon}}_{ij}{{\upkappa}}_{kl}.
\end{align}
The needed constitutive relations for the stress tensor ${{\upsigma}}_{ij}$ and the couple stress tensor ${{\upmu}}_{ij}$,
supplementing the equations of equilibrium \eqref{microeqlb}, are
\begin{align}
{{\upsigma}}_{ij}=\frac{{\pd} {{\mathcal{E}}}}{{\pd} {{\upepsilon}}_{ij}}={{A}}_{ijkl} {{\upepsilon}}_{kl}+{{D}}_{ijkl}{{\upkappa}}_{kl},\quad
{{\upmu}}_{ij}=\frac{{\pd} {{\mathcal{E}}}}{{\pd} {{\upkappa}}_{ij}}={{B}}_{ijkl}{{\upkappa}}_{kl}+{{D}}_{klij}{{\upepsilon}}_{kl},
\label{sigmmunc}
\end{align}
where ${{A}}_{ijkl}$, ${{D}}_{ijkl}$ and ${{B}}_{ijkl}$ are the (constant) elastic tensors of micropolar elasticity with the symmetries of ${{\mathbb{A}}}$ and ${{\mathbb{B}}}$ given by
\begin{align}
\label{symm}
{{A}}_{ijkl}={{A}}_{klij},\quad
{{B}}_{ijkl}={{B}}_{klij}.
\end{align}

It is easy to see that the Euler--Lagrange equations for the functional \eqref{en-grad1} are
\begin{align}
{{A}}_{ijkl}\frac{{\pd}} {{\pd}x_j}({u}_{k,l}+{{\mathtt{E}}}_{mkl}{\phi}_m)+{{D}}_{ijkl}\frac{{\pd}} {{\pd}x_j}{\phi}_{k,l}=0,
\end{align}
and
\begin{align}
{{B}}_{ijkl}\frac{{\pd}}{{\pd}x_j}{\phi}_{k,l}+{{D}}_{klij}\frac{{\pd}}{{\pd}x_j}({u}_{k,l}+{{\mathtt{E}}}_{mkl}{\phi}_m)-{{\mathtt{E}}}_{ijk} ({{A}}_{jkmn} ({u}_{m,n}+{{\mathtt{E}}}_{pmn}{\phi}_p)+{{D}}_{jkmn}{\phi}_{m,n})=0,
\end{align}
which naturally coincide with the equations of equilibrium \eqref{microeqlb} in the absence of external forces and couples.
Expanding further,
\begin{align}
{{A}}_{ijkl}{u}_{k,lj}+{{D}}_{ijkl}{\phi}_{k,lj}+{{A}}_{ijkl}{{\mathtt{E}}}_{mkl}{\phi}_{m,j}=0,
\end{align}
and
\begin{align}
{{B}}_{ijkl}{\phi}_{k,lj}+{{D}}_{klij}{u}_{k,lj}+{{D}}_{klin}{{\mathtt{E}}}_{mkl}{\phi}_{m,n}-{{\mathtt{E}}}_{ijk} ({{A}}_{jkmn} ({u}_{m,n}+{{\mathtt{E}}}_{pmn}{\phi}_p)+{{D}}_{jkmn}{\phi}_{m,n})=0.
\end{align}
We intend to look for null Lagrangian of the form of \eqref{en-grad2}.
For the conditions on such a null Lagrangian, 
we insist on the vanishing of the highest order derivative term (second order) in above equation. In other words, a set of sufficient conditions for the null Lagrangian are
\begin{align}
\label{symm1}
{{A}}_{ijkl}+{{A}}_{ilkj}=0, \quad {{B}}_{ijkl}+{{B}}_{ilkj}=0, \quad {{D}}_{ijkl}+{{D}}_{ilkj}=0,
\end{align}
and
\begin{align}
\label{symm2}
{{\mathtt{E}}}_{mkl}{{D}}_{klin}-{{\mathtt{E}}}_{ijk} {{D}}_{jkmn}=0,
\end{align}
as well as,
\begin{align}
\label{symm3}
{{\mathtt{E}}}_{mkl}{{A}}_{ijkl}=0, \quad
{{\mathtt{E}}}_{ijk}{{A}}_{jkmn}=0, \quad
{{\mathtt{E}}}_{ijk}{{A}}_{jkmn} {{\mathtt{E}}}_{pmn}=0.
\end{align}
In view of \eqref{symm}${}_1$, it is clear that the last three conditions are implied by ${{\mathtt{E}}}_{mkl}{{A}}_{ijkl}=0$. This implies, for $k\ne l$,
\begin{align}
\label{symm4}
{{A}}_{ijkl}-{{A}}_{ijlk}=0,
\end{align}
however, due to \eqref{symm}${}_1$, for $i\ne j$: ${{A}}_{ijkl}-{{A}}_{jikl}=0$ and ${{A}}_{ijkl}+{{A}}_{kjil}=0$.
This leads to
\begin{align}
\label{symm5}
{{A}}_{ijkl}=-{{A}}_{kjil}=-{{A}}_{jkil}={{A}}_{ikjl}.\\
\therefore\quad {{A}}_{ijkl}=-{{A}}_{kjil}=-{{A}}_{jkil}=-{{A}}_{jikl}=-{{A}}_{ijkl}\\
\implies {{A}}_{ijkl}=0\implies{\mathbb{A}}=0.
\end{align}

Concerning the conditions for $\mathbb{D},$ we make the following claim:
\begin{claim}
The conditions \eqref{symm1}${}_3$ and \eqref{symm2} are equivalent to the following
\begin{subequations}
\begin{align}
{{D}}_{ijkl}&=-{{D}}_{ilkj}\\
{{D}}_{ijji}&=-{{D}}_{ikki}\text{ for }i\ne j\ne k\ne i\\
{{D}}_{ijjk}&={{D}}_{kikk}+{{D}}_{jijk}\text{ for }i\ne j\ne k\ne i.
\end{align}
\label{symmset1}
\end{subequations}
\label{claim1}
\end{claim}
\begin{proof}
First assume that \eqref{symm1}${}_3$ and \eqref{symm2} hold.
By \eqref{symm1}${}_3$,
\begin{align}
{{D}}_{riki}=0\label{E:pex}
\end{align}
Put $(h,s,t)=\pi(1,2,3)$. Then \eqref{symm2} gives (no sum)
\begin{gather}
{{D}}_{thht}-{{D}}_{htht}={{D}}_{stst}-{{D}}_{tsst}\notag\\
\implies\;{{D}}_{thht}=-{{D}}_{tsst}\quad\text{by \eqref{E:pex}}\label{E:pff}
\end{gather}
Also if $(h,s,t)=-\pi(1,2,3)$ then by same calculation \eqref{E:pff} holds.\\
We now put $s=t\ne h$, and assume that $(s,h,k)=\pi(1,2,3)$. Then \eqref{symm2} gives (no sum)
\begin{gather}
{{D}}_{hkhs}-{{D}}_{khhs}={{D}}_{ksss}-{{D}}_{skss}\notag\\
\implies\;{{D}}_{khhs}={{D}}_{hkhs}+{{D}}_{skss}\quad\text{by \eqref{E:pex}}\label{E:ps}
\end{gather}
Also if $(s,h,k)=-\pi(1,2,3)$ then by same calculation \eqref{E:ps} holds. Also for $h=t\ne s$, by similar argument \eqref{E:ps} holds.

For the proof in the other direction, we assume that \eqref{symmset1} holds.
To obtain \eqref{symm2}, we fix $h\in\{1,2,3\}$. If $k=h$ then \eqref{symm2} trivially holds, so that we consider $k\ne h$. Then
\[
\sum_{i,j=1}^3{{D}}_{ijkl}e_{ijh}=\left({{D}}_{ijkl}-{{D}}_{jikl}\right)e_{ijh}\;\text{(no\;sum)}
\]
where $h$ is fixed, $i$ and $j$ are also fixed considering non-zero contributions. Then six cases arises for $k$ and $l$ ($k$ has two values and $l$ has three). Let us divide these six possibilities in the following three cases:\\
\textbf{Case 1.} [$h\ne l\ne k$]\\
This case consists of two sub-cases $(h,k,l)=\pi(1,2,3)$ and $(h,k,l)=-\pi(1,2,3)$. Then for nonzero term, $\{i,j\}=\{k,l\}$. Without any loss of generality, we assume that $i=l$. Therefore $\sum_{i,j=1}^3{{D}}_{ijkl}e_{ijh}={{D}}_{ijkl}e_{ijh}={{D}}_{lkkl}e_{lkh}$. Similarly, $\sum_{i,j=1}^3{{D}}_{ijhl}e_{ijk}={{D}}_{lhhl}e_{lhk}.$ But ${{D}}_{lkkl}e_{lkh}=-{{D}}_{lhhl}e_{lkh}={{D}}_{lhhl}e_{lhk},$ so that \eqref{symm2} holds.\\
\textbf{Case 2.} [$k=l$]\\
This case consists of two sub-cases $(p,l,h)=\pi(1,2,3)$ and $(p,l,h)=-\pi(1,2,3)$, where $p$ is the only element of the set $\{1,2,3\}\backslash\{h,k\}$. Then for nonzero contributions of the relevant term, either $i=k=l$ or $j=k=l$. Without any loss of generality, we assume that $i=k=l$. Therefore $\sum_{i,j=1}^3{{D}}_{ijkl}e_{ijh}={{D}}_{ijkl}e_{ijh}={{D}}_{ljll}e_{ljh}=\left({{D}}_{jhhl}-{{D}}_{hjhl}\right)e_{ljh}=\left({{D}}_{hjhk}-{{D}}_{jhhk}\right)e_{hjk}$. Thus, for $k=l$, 
\[
\sum_{i,j=1}^3{{D}}_{ijhl}e_{ijk}=\left({{D}}_{hjhl}-{{D}}_{jhhl}\right)e_{hjk}(WLOG)=\left({{D}}_{hjhk}-{{D}}_{jhhk}\right)e_{hjk}=\sum_{i,j=1}^3{{D}}_{ijkl}e_{ijh}.
\]
Therefore, \eqref{symm2} holds.\\
\textbf{Case 3.} [$h=l]$\\
This case consists of two sub-cases $(p,k,l)=\pi(1,2,3)$ and $(p,k,l)=-\pi(1,2,3)$, where $p$ is the only element of the set $\{1,2,3\}\backslash\{h,k\}$. Similar calculation as case 2 shows that \eqref{symm2} holds.\\
\end{proof}

\begin{remark}
According to above, it is noteworthy that \eqref{symm1}${}_1$, \eqref{symm1}${}_2$, and \eqref{symmset1} are sufficient conditions on the fourth order tensors ${{\mathbb{A}}}$, ${{\mathbb{B}}}$, and ${{\mathbb{D}}}$ satisfying the symmetries \eqref{symm} for an energy function to be a null Lagrangian in the linear micropolar elastostatics.
\end{remark}

\subsection{Centrosymmetric model}

In the absence of chiral term involving ${\mathbb{D}}$, for the general anisotropic case of the so called centrosymmetric model, the constitutive relations \eqref{sigmmunc} become
${{\upsigma}}_{ij}={{A}}_{ijkl}{{\upepsilon}}_{kl},
{{\upmu}}_{ij}={{B}}_{ijkl}{{\upkappa}}_{kl}.$
The suitable expression of the stored energy function for centrosymmetric micropolar material is
\begin{equation}
\label{E:ef}
{{\mathcal{E}}}({{\upepsilon}},{{\upkappa}})=\frac{1}{2}{{\upepsilon}}{\cdot} {{\mathbb{A}}}[{{\upepsilon}}]+\frac{1}{2}{{\upkappa}}{\cdot} {{\mathbb{B}}}[{{\upkappa}}].
\end{equation}
Thus the total stored energy functional is
\begin{equation}
\label{E:teC}
{{\mcl{V}}}({{{\bsym{u}}}},{{\bsym{\phi}}})=\frac{1}{2}\int_\Omega({{\upepsilon}}{\cdot} {{\mathbb{A}}}[{{\upepsilon}}]+{{\upkappa}}{\cdot} {{\mathbb{B}}}[{{\upkappa}}])dv.
\end{equation}
By 
\eqref{E:akr},
the Lagrangian 
\begin{align}
\label{E:lg}
{\Psi}
(\nabla{{{\bsym{u}}}},\nabla{{\bsym{\phi}}})=\frac{1}{2}({u}_{i,j}+{{\mathtt{E}}}_{ijs}{\phi}_s){{A}}_{ijkl}({u}_{k,l}+{{\mathtt{E}}}_{klr}{\phi}_r)
+\frac{1}{2}{\phi}_{i,j}{{B}}_{ijkl}{\phi}_{k,l}
\end{align}
is a polynomial of degree $2$.
It is of interest to characterize the possible null Lagrangians within the family stated by H. Rund \cite{Rund66}.

As a departure from the indicial notation used so far, we employ a different notation in the following. This closely follows the covariant/contravariant notation, but it should be remembered that we are not dealing with a curvilinear basis here; in fact we stick to the standard basis $\{{\bsym{e}}_1, {\bsym{e}}_2, {\bsym{e}}_3\}$ of ${{\mathbb{R}}}^3$ (and $\{{\bsym{e}}_1, {\bsym{e}}_2, \dotsc, {\bsym{e}}_N\}$ of ${{\mathbb{R}}}^N$) and employ $\{{\bsym{e}}^1, {\bsym{e}}^2, {\bsym{e}}^3\}$ as its own reciprocal copy (resp. $\{{\bsym{e}}^1, {\bsym{e}}^2, \dotsc, {\bsym{e}}^N\}$ of ${{\mathbb{R}}}^N$).
In general, the Greek indices range over $1, 2, 3,$ while the Latin indices range over $1, \dotsc, N$ unless specified otherwise.
Accordingly, we denote the components of $\bsym{x}\in {{\mathbb{R}}}^3$ as $x^\alpha$.
In view of the analysis of other models in the sequel, it is useful to consider $\bsym{y}\in {{\mathbb{R}}}^N$; for instance, in the present section where we discuss linear micropolar media, we have $N=6$
with the identification
\begin{align}
\bsym{y}=({{\bsym{u}}},{\bsym{\phi}}), \quad D{\bsym{y}}=(\nabla{{{\bsym{u}}}},\nabla{{\bsym{\phi}}}).
\end{align}
The components of $\bsym{y}$ are denoted by $y^i$, and the components of its derivative $D{\bsym{y}}\in {{\mathbb{R}}}^{N\times3}$ as $y^i_\alpha$.
Using this notation, the Euler operator for an arbitary Lagrangian ${\Psi}$ is found to be \cite{Rund66}
\begin{equation}
\label{E:ele}
{\mathscr{E}}_k({\Psi})=\frac{{\pd}^2{\Psi}}{{\pd} x^{{{\gamma}}}{\pd} y^k_{{{\gamma}}}}+\frac{{\pd}^2{\Psi}}{{\pd} y^j{\pd} y^k_{{{\gamma}}}}y^j_{{{\gamma}}}+\frac{{\pd}^2{\Psi}}{{\pd} y^j_{{{\beta}}}{\pd} y^k_{{{\gamma}}}}\frac{{\pd}^2y^j}{{\pd} x^{{{\beta}}}{\pd} x^{{{\gamma}}}}-\frac{{\pd}{\Psi}}{{\pd} y^k},
\end{equation}
where $k=1,\dotsc,N$. 
Here we recall that a Lagrangian 
is 
called null Lagrangian if and only if
\begin{equation}
\label{E:eqv}
{\mathscr{E}}_k({\Psi})\equiv0.
\end{equation}

The characterization theorem for null Lagrangians of the form of ${\Psi}$ ($:={\Psi}(D{\bsym{y}})$) 
is stated as Theorem \ref{mainthm} in Appendix \ref{S:cnl}; In fact,
${\Psi}$ is a null Lagrangian of polynomial degree $2$ in $D{\bsym{y}}$ if 
\begin{equation}
\label{E:flg}
{\Psi}
=\frac{1}{2}\mfk{D}(\alpha_1,\alpha_2;i_1,i_2)y^{i_1}_{\alpha_1}y^{i_2}_{\alpha_2}
+\mfk{D}(\alpha_1;i_1)y^{i_1}_{\alpha_1}
+\frac{1}{2}\mfk{D}(0;0),
\end{equation}
where $S^\alpha(\bsym{x},\bsym{y})$ 
are $\mathcal{C}^2$ functions,
while
\begin{align}
\mfk{D}(\alpha_1,\alpha_2;i_1i_2)&:=
\begin{vmatrix}
S^{\alpha_1}_{i_1}&S^{\alpha_2}_{i_1}\\
S^{\alpha_1}_{i_2}&S^{\alpha_2}_{i_2}
\end{vmatrix}, \quad
\mfk{D}(\alpha_1;i_1)&:=
\begin{vmatrix}
S^{\alpha_1}_{i_1}&S^{\alpha_2}_{i_1}\\
S^{\alpha_1}_{|\alpha_2}&S^{\alpha_2}_{|\alpha_2}
\end{vmatrix}, \quad
\mfk{D}(0;0)&:=
\begin{vmatrix}
S^{\alpha_1}_{|\alpha_1}&S^{\alpha_2}_{|\alpha_1}\\
S^{\alpha_1}_{|\alpha_2}&S^{\alpha_2}_{|\alpha_2}
\end{vmatrix}.
\label{Daaii}
\end{align}
For the purpose of comparison with \eqref{E:lg}, it is useful to consider above in terms of splitting of $\bsym{y}$ in ${{\bsym{u}}}$ and ${\bsym{\phi}}$. Thus, \eqref{E:flg} can be expressed as
\begin{align}
{\Psi}(\nabla{{{\bsym{u}}}},\nabla{{\bsym{\phi}}})
&=\frac{1}{2}\mfk{D}_1(\alpha_1,\alpha_2;i_1,i_2){u}_{i_1,\alpha_1}{u}_{i_2,\alpha_2}+\frac{1}{2}\mfk{D}_2(\alpha_1,\alpha_2;i_1,i_2){\phi}_{i_1',\alpha_1}{\phi}_{i_2',\alpha_2}+\notag\\
&\frac{1}{2}\mfk{D}_3(\alpha_1,\alpha_2;i_1,i_2){u}_{i_1,\alpha_1}{\phi}_{i_2',\alpha_2}+\mfk{D}_1(\alpha_1;i_1){u}_{i_1,\alpha_1}+\mfk{D}_2(\alpha_1;i_1){\phi}_{i_1',\alpha_1}\notag\\
&+\frac{1}{2}\mfk{D}(0;0)\label{E:flgs}
\end{align}
where 
$\mfk{D}_1$'s are restriction of corresponding $\mfk{D}$'s \eqref{Daaii} for $1\le i_1,i_2\le3$,
$\mfk{D}_2$'s for $4\le i_1,i_2\le6$,
$\mfk{D}_3$ for $1\le i_1\le3,\,4\le i_2\le6$, while the primed indices are given by
\[
i_1'=i_1-3,\,i_2'=i_2-3.
\]
The representation for $\mfk{D}_3$ in \eqref{E:flgs} is valid since $\mfk{D}_3(\alpha_1,\alpha_2;i_1,i_2)=\mfk{D}_3(\alpha_2,\alpha_1;i_2,i_1).$

With the expanded form \eqref{E:flgs} of the expression of stored energy \eqref{E:flg}, it is easier to compare it with the form prescribed by \eqref{E:lg}.
Thus, \eqref{E:lg} is a null Lagrangian provided there exist functions $S^\alpha$ such that
\begin{alignat}{8}
\begin{vmatrix}
S^{\alpha_1}_{i_1}&S^{\alpha_2}_{i_1}\\
S^{\alpha_1}_{i_2}&S^{\alpha_2}_{i_2}
\end{vmatrix}
&=\mfk{D}_1(\alpha_1,\alpha_2;i_1,i_2)={{A}}_{i_1\alpha_1i_2\alpha_2},\quad&
\begin{vmatrix}
S^{\alpha_1}_{i_1}&S^{\alpha_2}_{i_1}\\
S^{\alpha_1}_{i_2}&S^{\alpha_2}_{i_2}
\end{vmatrix}
&=\mfk{D}_2(\alpha_1,\alpha_2;i_1,i_2)={{B}}_{i_1'\alpha_1i_2'\alpha_2},\label{E:ics}\\
\begin{vmatrix}
S^{\alpha_1}_{i_1}&S^{\alpha_2}_{i_1}\\
S^{\alpha_1}_{i_2}&S^{\alpha_2}_{i_2}
\end{vmatrix}
&=\mfk{D}_3(\alpha_1,\alpha_2;i_1,i_2)=0,\quad&
\begin{vmatrix}
S^{\alpha_1}_{i_1}&S^{\alpha_2}_{i_1}\\
S^{\alpha_1}_{|\alpha_2}&S^{\alpha_2}_{|\alpha_2}
\end{vmatrix}
&=\mfk{D}_1(\alpha_1;i_1)=e_{ijs}{\phi}_s{{A}}_{iji_1\alpha_1},\label{E:icfr}\\
\begin{vmatrix}
S^{\alpha_1}_{i_1}&S^{\alpha_2}_{i_1}\\
S^{\alpha_1}_{|\alpha_2}&S^{\alpha_2}_{|\alpha_2}
\end{vmatrix}
&=\mfk{D}_2(\alpha_1;i_1)=0,\quad&
\begin{vmatrix}
S^{\alpha_1}_{|\alpha_1}&S^{\alpha_2}_{|\alpha_1}\\
S^{\alpha_1}_{|\alpha_2}&S^{\alpha_2}_{|\alpha_2}
\end{vmatrix}
&=\mfk{D}(0;0)=e_{ijs}e_{klr}{\phi}_s{\phi}_r{{A}}_{ijkl}.
\label{E:icsx}
\end{alignat}

We obtain further explicit form of necessary conditions on ${{\mathbb{A}}}$ and ${\mathbb{B}}$ from these conditions \eqref{E:ics}--\eqref{E:icsx}. From 
\eqref{E:ics}, using the properties of the determinant, we get the following restrictions for ${\mathbb{A}}$ and ${\mathbb{B}}$:
\begin{align}
{{A}}_{ijkl}={{A}}_{klij},\quad&{{B}}_{ijkl}={{B}}_{klij},\label{E:cndf}\\
{{A}}_{ilkj}=-{{A}}_{ijkl},\quad&{{B}}_{ilkj}=-{{B}}_{ijkl},\label{E:cnds}\\
{{A}}_{kjil}=-{{A}}_{ijkl},\quad&{{B}}_{kjil}=-{{B}}_{ijkl},\label{E:cndt}\\
{{A}}_{ijkl}=0,\quad&{{B}}_{ijkl}=0,&\text{if $i=k$ or $j=l$}.\label{E:cndl}
\end{align}
In order to proceed towards more simplification, let us introduce an alternate working notation for the functions of the form $S^{\alpha_1}_{i_1}$ and bring in more explicit form of the partial derivative involved. In view of the partition in the range of index $i_1$, etc, for clearer exposition, it is convenient
to denote 
\[
\frac{\partial S^\alpha}{\partial {u}_i}\text{ as }S^\alpha_{{u}_i}\text{ and }\frac{\partial S^\alpha}{\partial {\phi}_i}\text{ as }S^\alpha_{{\phi}_i},
\]
where $\alpha$ and $i$ both range over $1,2,3.$ Let us consider the following cases:\\
\textbf{Case 1.} Let 
$S^\alpha_{{u}_i}\ne0\ne S^\alpha_{{\phi}_i}\;\forall\alpha,i\in\{1,2,3\}.$
Then as 
\eqref{E:icfr}${}_1$ holds, i.e., 
$\begin{vmatrix}S^{\alpha_1}_{i_1}&S^{\alpha_2}_{i_1}\\
S^{\alpha_1}_{i_2}&S^{\alpha_2}_{i_2}\end{vmatrix}=0\Leftrightarrow
\begin{vmatrix}
S^{\alpha_1}_{{u}_i}&S^{\alpha_2}_{{u}_i}\\
S^{\alpha_1}_{{\phi}_j}&S^{\alpha_2}_{{\phi}_j}
\end{vmatrix}=
0$. Hence,
$\exists_{c_{ij}}$ such that
$0\ne|c_{ij}|<\infty$
and
\[
S^\alpha_{{u}_i}=c_{ij}S^\alpha_{{\phi}_j}\quad\text{(no sum over $j$).}
\]
This implies
\begin{align*}
{{A}}_{i\alpha j\beta}&=
\begin{vmatrix}
S^\alpha_{{u}_i}&S^\beta_{{u}_i}\\
S^\alpha_{{u}_j}&S^\beta_{{u}_j}
\end{vmatrix}
=c_{jj}
\begin{vmatrix}
S^\alpha_{{u}_i}&S^\beta_{{u}_i}\\
S^\alpha_{{\phi}_j}&S^\beta_{{\phi}_j}
\end{vmatrix}
=c_{jj}\mfk{D}_3(\alpha,\beta;i,j)=0\\
\text{and }{{B}}_{i\alpha j\beta}&=
\begin{vmatrix}
S^\alpha_{{\phi}_i}&S^\beta_{{\phi}_i}\\
S^\alpha_{{\phi}_j}&S^\beta_{{\phi}_j}
\end{vmatrix}
=\frac{1}{c_{ii}}
\begin{vmatrix}
S^\alpha_{{u}_i}&S^\beta_{{u}_i}\\
S^\alpha_{{\phi}_j}&S^\beta_{{\phi}_j}
\end{vmatrix}
=\frac{1}{c_{ii}}\mfk{D}_3(\alpha,\beta;i,j)=0
\end{align*}
Therefore in this case no non-trivial null Lagrangian exists.\\
\textbf{Case 2.} Let for some $\tilde{\alpha},\tilde{i}\in\{1,2,3\},\;S^{\tilde{\alpha}}_{{\phi}_{\tilde{i}}}=0$.
Then as 
\eqref{E:icfr}${}_1$
holds, it implies $\exists_{a_{\tilde{i}},a_j}$ such that
not both are zero, $|a_{\tilde{i}}|,|a_j|<\infty$
and
\[
\begin{cases}
a_{\tilde{i}}S^{\tilde{\alpha}}_{{\phi}_{\tilde{i}}}=a_jS^{\tilde{\alpha}}_{{u}_j}\\
a_{\tilde{i}}S^\beta_{{\phi}_{\tilde{i}}}=a_jS^\beta_{{u}_j}
\end{cases}
\quad\text{(no sum over ${\tilde{i}},j$).}
\]
If $a_{\tilde{i}}=0$, then $a_j\ne0$ which implies $S^{\tilde{\alpha}}_{{u}_j}=0=S^\beta_{{u}_j}$. 
If $a_{\tilde{i}}\ne0$ then $a_j=0$ implies $S^\beta_{{\phi}_{\tilde{i}}}=0$ and $S^{\tilde{\alpha}}_{{u}_j}=0$. 
Therefore,
$\mfk{D}_3({\tilde{\alpha}},\beta;j,{\tilde{i}})=0\implies S^{\tilde{\alpha}}_{{u}_j}=0\;or\;S^\beta_{{\phi}_{\tilde{i}}}=0$.

Consider the following subcases:

\textbf{subcase 2.1.} Let $S^\beta_{{\phi}_{\tilde{i}}}\ne0$ for some $\beta\ne{\tilde{\alpha}}$. 
Then $\mfk{D}_3({\tilde{\alpha}},\beta;j,{\tilde{i}})=0$ implies $S^{\tilde{\alpha}}_{{u}_j}=0\;\forall j\in\{1,2,3\}$. i.e., $S^{\tilde{\alpha}}$ is independent of ${{\bsym{u}}}$. 
Therefore ${{A}}_{k\gamma l\eta}=0$ if any one of $\gamma,\eta$ equals ${\tilde{\alpha}}$.

\textbf{subcase 2.2.} Let $S^{\tilde{\alpha}}_{{u}_j}\ne0$ for some ${\tilde{\alpha}},j\in\{1,2,3\}$. 
Then $\mfk{D}_3({\tilde{\alpha}},\beta;j,{\tilde{i}})=0$ implies $S^\beta_{{\phi}_{\tilde{i}}}=0\;\forall\beta\in\{1,2,3\}.$ i.e., $S^\beta$ is independent of ${\phi}_{\tilde{i}}$ for all $\beta$. 
Therefore ${{B}}_{k\gamma l\eta}=0$ if any one of $k,l$ equals ${\tilde{i}}$.

Thus we conclude the following:

\quad 1. If all $S^{\alpha}$ depends on all components of ${{\bsym{u}}}$ and ${\phi}_i$ then
$
{\mathbb{A}}=0={\mathbb{B}}.
$

\quad 2. If $S^{\alpha}_{{\phi}_i}=0$ for some ${\alpha}, i\in\{1,2,3\}$ then $S^{\alpha}$ is independent of ${{\bsym{u}}}$ or $S^\beta$ is independent of ${\phi}_i$ for all $\beta$.

\quad 3. 
By interchanging the role of ${{\bsym{u}}}$ and ${\bsym{\phi}}$ in \textbf{case 2} thus far,
if $S^{\alpha}_{{u}_i}=0$ for some ${\alpha},i\in\{1,2,3\}$ then $S^{\alpha}$ is independent of ${\bsym{\phi}}$ or $S^\beta$ is independent of ${u}_i$ for all $\beta$.

Under \textbf{subcase 2.1}, consider two subcases as follows:

\textbf{subcase 2.1.1.} All other partial derivatives are non-zero. Then By first conclusion, ${{A}}_{i\gamma j\eta}=0$ whenever $\gamma\ne{\tilde{\alpha}}\ne\eta$. Therefore by the result in \textbf{subcase 2.1}, ${\mathbb{A}}=0$.

\textbf{subcase 2.1.2.} Any one of other partial derivatives is zero. Then we can apply second or third conclusion. Thus applying the conclusions repeatedly, after a finite number of steps we find that either ${\mathbb{A}}=0={\mathbb{B}}$ 
or $S^{\alpha}$ is independent of ${{\bsym{u}}}$ for all ${\alpha}$ 
or $S^{\alpha}$ is independent of ${\bsym{\phi}}$ for all ${\alpha}$.
When all $S^{\alpha}$'s are independent of ${{\bsym{u}}}$, by \eqref{E:ics}${}_1$, ${\mathbb{A}}=0$. And when all $S^{\tilde{\alpha}}$'s are independent of ${\bsym{\phi}}$, by \eqref{E:ics}, ${\mathbb{B}}=0$. Our claim is in this case also ${\mathbb{A}}=0$.
To prove the claim, let us choose ${\bsym{\phi}}\equiv0$. Then RHS (right hand side) of \eqref{E:icfr}${}_2$ becomes {\em zero} and hence LHS (left hand side). Since LHS determinants do not depend on ${\bsym{\phi}}$, they are {\em zero} for all ${\bsym{\phi}}$. Therefore we get
\begin{align}
e_{ijs}{\phi}_s{{A}}_{ijkl}&=0\label{E:icsv}
\end{align}
We choose ${\bsym{\phi}}(\bsym{x})=({\phi}_1(\bsym{x}),{\phi}_2(\bsym{x}),{\phi}_3(\bsym{x}))=(1,0,0)$. Then from \eqref{E:icsv},
\begin{align*}
e_{321}{{A}}_{32kl}+e_{231}{{A}}_{23kl}=0&\quad1\le k,l\le3\\
\implies {{A}}_{23kl}={{A}}_{32kl}
\end{align*}
Similarly taking different ${\bsym{\phi}}$'s we get
${{A}}_{ijkl}={{A}}_{jikl},$ and
${{A}}_{ijkl}=-{{A}}_{kjil}=-{{A}}_{jkil}={{A}}_{ikjl}.$
Therefore $-{{A}}_{ijkl}={{A}}_{kjil}={{A}}_{kijl}={{A}}_{ikjl}={{A}}_{ijkl},$
$\implies{{A}}_{ijkl}=0, \text{i.e.,}{\mathbb{A}}=0.$

Thus for all cases it is necessary that 
\begin{align}
{\mathbb{A}}=0.
\label{E:cndAzero}\end{align}
Clearly, it is not necessary that ${\mathbb{B}}$ is $0$. 
Thus the necessary 
conditions for null Lagrangian of the form \eqref{E:lg}, in centrosymmetric case (note that we assume ${{B}}_{klij}={{B}}_{ijkl}$) {\em and within the family stated by H. Rund \cite{Rund66}}, are
${{B}}_{ilkj}=-{{B}}_{ijkl}, {{B}}_{kjil}=-{{B}}_{ijkl}, \forall i, j, k, l,$ and ${{B}}_{ijkl}=0,$ if $i=k$ or $j=l$, and, most importantly, ${{\mathbb{A}}}=0.$
It is emphasized that the last condition completely removes any contribution from the displacement field in a possible expression for a null Lagrangian.

\begin{remark}
Although we briefly the same in the context of Claim \ref{claim1}, the sufficiency of \eqref{E:cndAzero} and \eqref{E:cnds}${}_2$, so that \eqref{E:lg} is a null Lagrangian, can be easily verified by a direct calculation. Let
\begin{align*}
F&:=\frac{1}{2}{\phi}_{i,j}{{B}}_{ijkl}{\phi}_{k,l}.
\end{align*}
Let $\mathscr E_i$ be the components of the Euler operator for $i=1,\dots,6.$ Then $\mathscr E_i(F)\equiv0$ for $i=1,2,3$ since $F$ is independent of ${{\bsym{u}}}$ and $\nabla{{\bsym{u}}}$. For $i=4,5,6$,
\[
\mathscr E_i(F)=\frac{\partial}{\partial x_j}\left(F_{{\phi}_{i,j}}\right),
\text{where }
F_{{\phi}_{i,j}}={{B}}_{ijkl}{\phi}_{k,l}.
\]
Therefore,
by 
\eqref{E:cnds}${}_2$,
\[
\frac{\partial}{\partial x_j}\left(F_{{\phi}_{i,j}}\right)=
{{B}}_{ijkl}{\phi}_{k,lj}=
\frac{1}{2}({{B}}_{ijkl}+{{B}}_{ilkj}){\phi}_{k,lj}=0.
\]
As $\mathscr E_k(F)\equiv0$ for $k=1,\dots,6$, hence $F$ is a null Lagrangian.
\end{remark}

\begin{remark}
In contrast to the case of linear elasticity above null Lagrangian can be interpreted as stored energy functionals in the linear micropolar theory.
\end{remark}

\subsection{Splitting of stored energy using null terms}

Let us decompose the fourth order elastic tensor ${{\mathbb{B}}}$ as
\[
{{\mathbb{B}}}=\hat{{\mathbb{B}}}+\tilde{{\mathbb{B}}}+\mathring{{\mathbb{B}}}
\]
\begin{align}
\text{where }
\hat{{B}}_{ijkl}:=\frac{1}{4}({{B}}_{ijkl}+{{B}}_{klij}+{{B}}_{ilkj}+{{B}}_{kjil}),\quad
\tilde{{B}}_{ijkl}:=\frac{1}{4}({{B}}_{ijkl}+{{B}}_{klij}-{{B}}_{ilkj}-{{B}}_{kjil}),
\end{align}
and
\begin{align}
\mathring{{B}}_{ijkl}=\frac{1}{2}({{B}}_{ijkl}-{{B}}_{klij}).
\end{align}
Then clearly $\tilde{{\mathbb{B}}}$ satisfies \eqref{E:cndf}--\eqref{E:cndl}. Due to
\eqref{symm}${}_2$, we have $\mathring{{\mathbb{B}}}=\mathbb{0}$
(which are total 45 conditions).
The remaining are 36 arbitrary constants in $\hat{{\mathbb{B}}}+\tilde{{\mathbb{B}}}$ and are split between $\hat{{\mathbb{B}}}$ and $\tilde{{\mathbb{B}}}.$ But
$\hat{{\mathbb{B}}}$ and $\tilde{{\mathbb{B}}}$ individually satisfy the relation that swaps the indices $(ij)$ and $(kl)$.
By symmetry,
there are $9+18\times2=45$ zero entries in $\tilde{{\mathbb{B}}}$.
Note that 
there are $18$ independent, possibly non-zero, entries in $\tilde{{\mathbb{B}}}$.
With this decomposition, the energy functional \eqref{E:teC} is split as follows:
\begin{align}
{{\mcl{V}}}&=\widehat{{{\mcl{V}}}}+\widetilde{{{\mcl{V}}}},\quad
\widehat{{{\mcl{V}}}}({{{\bsym{u}}}},{{\bsym{\phi}}}):=\frac{1}{2}\int_\Omega ({{\upepsilon}}{\cdot}{{\mathbb{A}}}[{{\upepsilon}}]+{{\upkappa}}{\cdot}\hat{{\mathbb{B}}}[{{\upkappa}}])dv,\label{E:slg}
\quad
\widetilde{{{\mcl{V}}}}({{{\bsym{u}}}},{{\bsym{\phi}}}):=\frac{1}{2}\int_\Omega {{\upkappa}}{\cdot}\tilde{{\mathbb{B}}}[{{\upkappa}}]dv,
\end{align}
where $\widehat{{{\mcl{V}}}}$ cannot be a null functional (i.e. with its Lagrangian as null) unless ${\mathbb{A}}=\mathbb{0}=\hat{{\mathbb{B}}}$. 
In this context, it is noted that the Lagrangian in \eqref{E:slg}${}_2$ can be null only if ${{\upepsilon}}{\cdot}{{\mathbb{A}}}[{{\upepsilon}}]=0={{\upkappa}}{\cdot}\hat{{\mathbb{B}}}[{{\upkappa}}]$, since if they are non-zero, ${{A}}_{ilkj}={{A}}_{ijkl}, 
\hat {{B}}_{ilkj}=\hat {{B}}_{ijkl}$ whch violate \eqref{E:cnds}.
Following \cite{Lancia95} (see also \cite{Carillo02}), the part $\tilde{{\mathbb{B}}}$ of the elasticity tensor affects equilibrium through boundary term only and the equilibrium equation on the boundary takes the form (recall \eqref{E:cauchy})
\begin{align*}
\mcl{M}[{\bsym{\phi}}]=\bsym{m}_0\text{ in $\partial\Omega$,}\quad
\text{ where }
\mcl{M}[{\bsym{\phi}}]&:={{\mathbb{B}}}[\nabla{\bsym{\phi}}]\bsym{n}=\hat{{\mathbb{B}}}[\nabla{\bsym{\phi}}]\bsym{n}+\tilde{\bsym{m}}[{\bsym{\phi}}],\quad
\tilde{\bsym{m}}[{\bsym{\phi}}]:=\tilde{{\mathbb{B}}}[\nabla{\bsym{\phi}}]\bsym{n},
\end{align*}
where $\bsym{n}$ is the outward unit normal on the boundary of the body $\partial\Omega$ while
$\bsym{m}_0$ is the externally applied surface moment (per unit area).
This leads to the existence of a surface potential \cite{Lancia95,Carillo02} (for moments) corresponding to $\tilde{\bsym{m}}[{\bsym{\phi}}]$, that is,
\begin{equation}\label{E:spt}
\tilde{\mcl{M}}[{\bsym{\phi}}]=\frac{1}{2}\int_{\partial\Omega}\tilde{{\mathbb{B}}}_{ijkl}{\phi}_{i,j}{\phi}_{k}n_ld\sigma.
\end{equation}
In order that $\tilde{{\mathbb{B}}}$ (that satisfies the symmetry \eqref{symm}${}_2$) vanishes we require that
$\tilde{{B}}_{ijkl}
=0;$
In other words,
\[
\begin{bmatrix}
\tilde{{B}}_{2233}&\tilde{{B}}_{1233}&\tilde{{B}}_{1322}\\
\tilde{{B}}_{2133}&\tilde{{B}}_{1133}&\tilde{{B}}_{2311}\\
\tilde{{B}}_{3122}&\tilde{{B}}_{3211}&\tilde{{B}}_{1122}\\
\end{bmatrix}=0,\quad
\begin{bmatrix}
\tilde{{B}}_{3113}&\tilde{{B}}_{3123}&\tilde{{B}}_{2132}\\
\tilde{{B}}_{3213}&\tilde{{B}}_{1221}&\tilde{{B}}_{1231}\\
\tilde{{B}}_{2312}&\tilde{{B}}_{1321}&\tilde{{B}}_{2332}\\
\end{bmatrix}=0,
\]
provide a set of $18$ additional symmetry relations. These can be interpreted as analogues to the six Cauchy's relations in the theory of linear elasticity \cite{cT94,Lo27,Lancia95,FI02,FI13}. 
\begin{remark}
In the above scenario ${\mathbb{B}}$ lacks any minor symmetries unless the space of rotation vectors is restricted to gradients of scalar functions (say, curl-free vector fields on simply connected domains so that ${\phi}_{i,j}={\phi}_{j,i}$). This is reflected in the fact that the number of conditions is 18 rather than the classical case of six \cite{cT94}.
The restriction to curl-free vector fields can be written as that when $\psi=0$ or $\chi=0$ in the Clebsch representation of ${\bsym{\phi}}$.
\end{remark}

\subsection{Non-chiral isotropic model}
In the case of an isotropic material, ${{\mathbb{A}}}$ and ${{\mathbb{B}}}$ take the form,
\[
\begin{cases}
{{A}}_{ijkl}={{\lambda}}\delta_{ij}\delta_{kl}+({{\mu}}+{{\kappa}})\delta_{ik}\delta_{jl}+{{\mu}}\delta_{il}\delta_{jk}~&={{A}}_{klij}\\
{{B}}_{ijkl}={{\upbeta_1}}\delta_{ij}\delta_{kl}+{{\upbeta_2}}\delta_{ik}\delta_{jl}+{{\upbeta_3}}\delta_{il}\delta_{jk}&={{B}}_{klij}
\end{cases}
\]
where ${{\mu}},{{\kappa}},{{\lambda}},{{\upbeta_1}},{{\upbeta_2}},{{\upbeta_3}}$ are constants. Therefore,
\begin{align*}
\quad{{\upsigma}}_{ij}&={{A}}_{ijkl}{{\upepsilon}}_{kl}=({{\mu}}+{{\kappa}}){{\upepsilon}}_{ij}+{{\mu}}{{\upepsilon}}_{ji}+{{\lambda}}\delta_{ij}{{\upepsilon}}_{ll},\\
\text{and 
}\quad{{\upmu}}_{ij}&={{B}}_{ijkl}{{\upkappa}}_{kl}={{\upbeta_2}}{{\upkappa}}_{ij}+{{\upbeta_3}}{{\upkappa}}_{ji}+{{\upbeta_1}}\delta_{ij}{{\upkappa}}_{ll}
\end{align*}
Using the fact that
\begin{equation}\label{E:pdse}{{\upepsilon}}{\cdot}{{\mathbb{A}}}[{{\upepsilon}}]
=({{\mu}}+{{\kappa}}){{\upepsilon}}{\cdot}{{\upepsilon}}+{{\mu}}{{\upepsilon}}{\cdot}{{\upepsilon}}^{\top}+{{\lambda}}(\tr{{\upepsilon}})^2\end{equation}
is positive, and similar condition for ${\upkappa}\cdot{{\mathbb{B}}}[{\upkappa}]$,
we get the well known \cite{Er66,EK76,Eringen99} necessary and sufficient conditions for positive definite stored energy are
\begin{equation}
{{\kappa}}>0, 2{{\mu}}+{{\kappa}}>0, 3{{\lambda}}+2{{\mu}}+{{\kappa}}>0, \quad
3{{\upbeta_1}}+{{\upbeta_2}}+{{\upbeta_3}}>0, {{\upbeta_2}}+{{\upbeta_3}}>0, {{\upbeta_2}}-{{\upbeta_3}}>0.
\end{equation}
From the perspective of null Lagrangian, we already have the restriction
\begin{equation}
\label{E:cif}
{{\mu}}={{\kappa}}={{\lambda}}=0,
\end{equation}
as ${{\mathbb{A}}}$ needs to be $0$ within the class of null Lagrangians due to H. Rund \cite{Rund66}.
Let us now simplify the expression involving ${\mathbb{B}}$.
Here
\begin{align}
\tilde{{B}}_{ijkl}&=\frac{1}{2}({{B}}_{ijkl}-{{B}}_{ilkj})
=\frac{1}{2}({{\upbeta_3}}-{{\upbeta_1}})(\delta_{jk}\delta_{il}-\delta_{ij}\delta_{kl}).
\end{align}
Therefore,
$\tilde{{\mathbb{B}}}[{\upkappa}]=\frac{1}{2}({{\upbeta_3}}-{{\upbeta_1}})[{\upkappa}^{\top}-(\tr{\upkappa})\mbf{I}]$,
so that
\begin{align}
{\upkappa}\cdot\tilde{{\mathbb{B}}}[{\upkappa}]&=\frac{1}{2}({{\upbeta_3}}-{{\upbeta_1}})\left[{\upkappa}\cdot{\upkappa}^{\top}-(\tr{\upkappa})^2\right].
\label{E:tda}
\end{align}
We denote the deviatoric part
\[
\dev{{\upkappa}}:={{\upkappa}}-\frac{1}{3}(\tr{{\upkappa}})\mbf{I}. 
\]
Then
$\hat{{{\upkappa}}}{\cdot}\hat{{{\upkappa}}}=(\dev\hat{{{\upkappa}}})^2+\frac{1}{3}(\tr{{\upkappa}})^2$
since $\tr\hat{{{\upkappa}}}=\tr{{\upkappa}}$ and $(\dev\hat{{{\upkappa}}}){\cdot}\mbf{I}
=0$. 
Also
\begin{align}
\hat{{{\upkappa}}}{\cdot}\tilde{{{\upkappa}}}
=0,\quad
{{\upkappa}}{\cdot}{{\upkappa}}
=(\dev\hat{{{\upkappa}}})^2+\frac{1}{3}(\tr{{\upkappa}})^2+\tilde{{{\upkappa}}}{\cdot}\tilde{{{\upkappa}}},
\text{ and }\quad{{\upkappa}}{\cdot}{{\upkappa}}^{\top}
=(\dev\hat{{{\upkappa}}})^2+\frac{1}{3}(\tr{{\upkappa}})^2-\tilde{{{\upkappa}}}{\cdot}\tilde{{{\upkappa}}}.
\label{E:gdgt}\end{align}
Therefore, we get\begin{align*}{{\upkappa}}{\cdot}{{\mathbb{B}}}[{{\upkappa}}]
&=(2{{\mu}}+{{\kappa}})(\dev\hat{{{\upkappa}}})^2+\frac{3{{\lambda}}+2{{\mu}}+{{\kappa}}}{3}(\tr{{\upkappa}})^2+{{\kappa}}\tilde{{{\upkappa}}}{\cdot}\tilde{{{\upkappa}}},
\end{align*} 
and
\begin{align}
{\upkappa}\cdot\hat{{\mathbb{B}}}[{\upkappa}]
&=\frac{2{{\upbeta_2}}+{{\upbeta_3}}+{{\upbeta_1}}}{2}(\dev\hat{{\upkappa}})^2+\frac{{{\upbeta_2}}+2{{\upbeta_3}}+2{{\upbeta_1}}}{3}(\tr{\upkappa})^2+\frac{2{{\upbeta_2}}-{{\upbeta_3}}-{{\upbeta_1}}}{2}\tilde{{\upkappa}}\cdot\tilde{{\upkappa}},
\label{Bhateqn}
\end{align}
by 
\eqref{E:gdgt}. For a null Lagrangian the non-trivial expression corresponds to ${\upkappa}\cdot{{\mathbb{B}}}[{\upkappa}]$, while its symmetric part leads to ${\upkappa}\cdot\hat{{B}}[{\upkappa}]=0$. Therefore, using \eqref{Bhateqn}, we have,
\begin{align*}
2{{\upbeta_2}}+{{\upbeta_3}}+{{\upbeta_1}}=0,\quad
{{\upbeta_2}}+2{{\upbeta_3}}+2{{\upbeta_1}}=0,\quad
2{{\upbeta_2}}-{{\upbeta_3}}-{{\upbeta_1}}=0
\end{align*}
which gives
\begin{equation}
\label{E:cis}
{{\upbeta_2}}=0\quad\text{and}\quad{{\upbeta_2}}-{{\upbeta_3}}-{{\upbeta_1}}=0
\end{equation}
and thus
\[
\boldsymbol\epsilon\cdot\tilde{\mathbb A}[\boldsymbol\epsilon]+{\upkappa}\cdot
{{\mathbb{B}}}[{\upkappa}]={\upkappa}\cdot\tilde{{\mathbb{B}}}[{\upkappa}]=({{\upbeta_2}}-{{\upbeta_3}})\left[(\tr{\upkappa})^2-{\upkappa}\cdot{\upkappa}^{\top}\right].
\]
Indeed, \eqref{E:cis} and \eqref{E:cif} are the necessary and sufficient conditions for null Lagrangian in isotropic case without chiral term ${\mathbb{D}}$ and within the class of null Lagrangians due to H. Rund \cite{Rund66}.
And the Lagrangian becomes
\begin{equation}
{\Psi}(\boldsymbol\epsilon,{\upkappa})=C_0\left[(\tr{\upkappa})^2-{\upkappa}\cdot{\upkappa}^{\top}\right]
\label{PsicentroIso}
\end{equation}
for some constant $C_0$.
\begin{remark}
Using the definition of the second invariant of a second order tensor, i.e.,
$\mcl{I}_2(\mbf{A})=\frac{1}{2}((\tr\mbf{A})^2-\tr\mbf{A}^2),$
we have an equivalent expression for \eqref{PsicentroIso}, i.e.,
${{\Psi}}({{\upepsilon}},{{\upkappa}})=C_0\mcl{I}_2({{\upkappa}}).$
It can be easily verified by direct calculation too that ${{\Psi}}$ is a null Lagrangian.
\end{remark}

\subsection{Chiral isotropic model}
Isotropy (which insists on a symmetry relative to all orthogonal tensors) implies centrosymmetry so the expressions discussed earlier may need to be modified with non-zero ${\mathbb{D}}$. In non-centrosymmetric case of symmetry similar to `Isotropy' (relative to all rotation tensors) does not include mirror-reflection, i.e., inversion, (sometimes called hemitropy) has additonal parameters.
In the case of general chiral `isotropic' 
micropolar elastic material, three additional material constants appear as compared to the non-chiral isotropic micropolar elastic material \cite{wN86,LakesB82,Lakes2001,
JOS11,NGS062}.
These additional material parameters allows the incorporation of a change in the signs of the corresponding terms depending on the handedness of the microstructure (owing to the presence of axial vector, rotation axis, cross product, etc). 
For hemitropic (chiral isotropic) material, ${{\mathbb{A}}}$, ${{\mathbb{B}}}$, ${{\mathbb{D}}}$ takes the form \cite{wJN77,wN86,Sharma2004,JOS11}
\begin{equation}
\label{E:nid}
\begin{cases}
{{A}}_{ijkl}={{\lambda}}\delta_{ij}\delta_{kl}+({{\mu}}+{{\kappa}})\delta_{ik}\delta_{jl}+{{\mu}}\delta_{il}\delta_{jk}&={{A}}_{klij}\\
{{B}}_{ijkl}={{\upbeta_1}}\delta_{ij}\delta_{kl}+{{\upbeta_2}}\delta_{ik}\delta_{jl}+{{\upbeta_3}}\delta_{il}\delta_{jk}&={{B}}_{klij}\\
{{D}}_{ijkl}={{\zeta}}\delta_{ij}\delta_{kl}+({{\nu}}+{{\rho}})\delta_{ik}\delta_{jl}+{{\nu}}\delta_{il}\delta_{jk}&={{D}}_{klij}.
\end{cases}
\end{equation}
In contrast to the classical isotropic linear elasticity which is characterized by two material constants, so called L\'{a}me, the non-centrosymmetric `isotropic' micropolar material possesses nine independent material constants.
The stress ${{\upsigma}}$ and couple stress ${{\upmu}}$ are given by \eqref{sigmmunc}.
For chiral isotropic case,
${{\upsigma}}_{ij}={{\lambda}}\delta_{ij}{{\upepsilon}}_{ll}+({{\mu}}+{{\kappa}}){{\upepsilon}}_{ij}+{{\mu}}{{\upepsilon}}_{ji}+{{\zeta}}\delta_{ij}{{\upkappa}}_{ll}+({{\nu}}+{{\rho}}){{\upkappa}}_{ij}+{{\nu}}{{\upkappa}}_{ji},
{{\upmu}}_{ij}={{\upbeta_1}}\delta_{ij}{{\upkappa}}_{ll}+{{\upbeta_2}}{{\upkappa}}_{ij}+{{\upbeta_3}}{{\upkappa}}_{ji}+{{\zeta}}\delta_{ij}{{\upepsilon}}_{ll}+({{\nu}}+{{\rho}}){{\upepsilon}}_{ij}+{{\nu}}{{\upepsilon}}_{ji}.$
For positive definite stored energy, 
\begin{align}
&{{\kappa}}>0,\;2{{\mu}}+{{\kappa}}>0,\;3{{\lambda}}+2{{\mu}}+{{\kappa}}>0,\label{E:npdf}\\
&3{{\upbeta_1}}+{{\upbeta_3}}+{{\upbeta_2}}>0,\;{{\upbeta_3}}+{{\upbeta_2}}>0,\;{{\upbeta_2}}-{{\upbeta_3}}>0.\label{E:npds}
\end{align}
According to \eqref{symm1}${}_1$ and \eqref{symm1}${}_3$, we require that
${{\zeta}}\delta_{ij}\delta_{kl}+({{\nu}}+{{\rho}})\delta_{ik}\delta_{jl}+{{\nu}}\delta_{il}\delta_{jk}
=-{{\zeta}}\delta_{kj}\delta_{il}-({{\nu}}+{{\rho}})\delta_{ki}\delta_{jl}-{{\nu}}\delta_{kl}\delta_{ji},$
so that
\begin{align}
{{\nu}}+{{\zeta}}={{\nu}}+{{\rho}}=0.
\end{align}
Hence, \eqref{E:nid}${}_3$ implies
${{D}}_{ijkl}={{\zeta}}(\delta_{ij}\delta_{kl}-\delta_{il}\delta_{jk}),$ so that
\begin{align}
{{\upepsilon}}{\cdot}{{\mathbb{D}}}[{{\upkappa}}]={{\zeta}}\left[(\tr{{\upepsilon}})(\tr{{\upkappa}})-{{\upepsilon}}{\cdot}{{\upkappa}}^{\top}\right].
\end{align}
Similarly,
\eqref{E:nid}${}_1$
and
\eqref{E:nid}${}_2$ implies 
\begin{align}
{{\mu}}+{{\lambda}}={{\mu}}+{{\kappa}}=0,
\text{ and }
{{\upbeta_3}}+{{\upbeta_1}}={{\upbeta_2}}=0,
\end{align}
so that
${{A}}_{ijkl}={{\lambda}}(\delta_{ij}\delta_{kl}-\delta_{il}\delta_{jk}),
{{B}}_{ijkl}={{\upbeta_1}}(\delta_{ij}\delta_{kl}-\delta_{il}\delta_{jk}),$
and
\begin{align}
{{\upepsilon}}{\cdot}{{\mathbb{A}}}[{{\upepsilon}}]={\lambda}\left[(\tr{{\upepsilon}})^2-{{\upepsilon}}{\cdot}{{\upepsilon}}^{\top}\right],\quad
{{\upkappa}}{\cdot}{{\mathbb{B}}}[{{\upkappa}}]={\upbeta_1}\left[(\tr{{\upkappa}})^2-{{\upkappa}}{\cdot}{{\upkappa}}^{\top}\right].
\end{align}
Also,
\begin{align*}
{{\upsigma}}_{ij}=-{{\lambda}}{{\upepsilon}}_{ji}+{{\lambda}}\delta_{ij}{{\upepsilon}}_{ll}-{{\zeta}}{{\upkappa}}_{ji}+{{\zeta}}\delta_{ij}{{\upkappa}}_{ll},\\
{{\upmu}}_{ij}=-{{\upbeta_1}}{{\upkappa}}_{ji}+{{\upbeta_1}}\delta_{ij}{{\upkappa}}_{ll}
-{{\zeta}}{{\upepsilon}}_{ji}+{{\zeta}}\delta_{ij}{{\upepsilon}}_{ll}.
\end{align*}
Let
\begin{align*}
F&:=\frac{1}{2}({u}_{i,j}+{{\mathtt{E}}}_{ijs}{\phi}_s){{A}}_{ijkl}({u}_{k,l}+{{\mathtt{E}}}_{kls}{\phi}_s)
+\frac{1}{2}{\phi}_{i,j}{{B}}_{ijkl}{\phi}_{k,l}
+({u}_{i,j}+{{\mathtt{E}}}_{ijs}{\phi}_s){{D}}_{ijkl}{\phi}_{k,l}.
\end{align*}
Let $\mathscr E_i$ be the components of the Euler operator for $i=1,\dots,6.$ 
Then for $i=1, 2, 3$, $\mathscr E_i(F)=0$ gives
\[
{{\upsigma}}_{ij,j}
=-{{\lambda}}{{\upepsilon}}_{ji,j}+{{\lambda}}{{\upepsilon}}_{ll,i}-{{\zeta}}{{\upkappa}}_{ji,j}+{{\zeta}}{{\upkappa}}_{ll,i}=0,
\]
and for $i=4, 5, 6$, $\mathscr E_i(F)=0$ gives
\[
{{\upmu}}_{ij,j}-{{\mathtt{E}}}_{ijk} {{\upsigma}}_{jk}
=-{{\upbeta_1}}{{\upkappa}}_{ji,j}+{{\upbeta_1}}{{\upkappa}}_{ll,i}-{{\zeta}}{{\upepsilon}}_{ji,j}+{{\zeta}}{{\upepsilon}}_{ll,i}=0,
\]
where ${{\upepsilon}}_{ll}={u}_{l,l}$ and ${{\upkappa}}_{ll}={\phi}_{l,l}$.
In fact,
\[
-{{\lambda}}({u}_{j,ij}+{{\mathtt{E}}}_{jis}{\phi}_{s,j})+{{\lambda}}{u}_{l,li}-{{\zeta}}{\phi}_{j,ij}+{{\zeta}}{\phi}_{l,li}=0, \quad
-{{\upbeta_1}}{\phi}_{j,ij}+{{\upbeta_1}}{\phi}_{l,li}-{{\zeta}}({u}_{j,ij}+{{\mathtt{E}}}_{jis}{\phi}_{s,j})+{{\zeta}}{u}_{l,li}=0,
\]
i.e.,
\[
{{\lambda}}{{\mathtt{E}}}_{jis}{\phi}_{s,j}=0, \quad
{{\zeta}}{{\mathtt{E}}}_{jis}{\phi}_{s,j}=0.
\]
As $\mathscr E_k(F)\equiv0$ for $k=1,\dots,6$, hence $F$ is a null Lagrangian if and only if ${\lambda}={\zeta}=0$.
Thus, we do not find any other null Lagrangians besides those which are isotropic.

\begin{remark}
On the other hand if the space of rotation vector field is restricted to curl free vector fields then the null Lagrangians with non-zero values of ${\lambda}$ and ${\zeta}$ are also admissible.
\end{remark}

\section{Generalized elasticity for quasicrystals}

We seek the null Lagrangians in the framework of the generalized elasticity theory of quasicrystals \cite{BL75,ADL92}.
In general, an $(n-3)$-dimensional quasicrystal can be generated by the projection of an $n$-dimensional periodic structure to the 3-dimensional physical space ($n=4,5,6$). Typically, the $n$-dimensional hyperspace $E^{n}$ is seen as the direct sum of two orthogonal subspaces (phonons and phasons, respectively), 
\begin{align}
\label{E-deco}
E^{n}=E_{{p}}^3\oplus E_{{s}}^{(n-3)}.
\end{align}
Let us restrict ourselves for illustrative purpose the case when $n=6$. 
Throughout this section, we denote the phonon fields by $(\cdot)^{{p}}$ and the phason fields by $(\cdot)^{{s}}$. It is important to note that all quantities (phonon and phason fields) depend on the referential (material) location ${\boldsymbol{x}} \in {{\mathbb{R}}}^3$.
In the theory of quasicrystals, the equilibrium conditions are of the form (see, e.g.,~\cite{Ding1993,Hu2000})
\begin{align}
{{\sigma}}^{{p}}_{ij,j}+f_i^{{p}}=0,\quad
{{\sigma}}^{{s}}_{ij,j}+f_i^{{s}}=0,
\label{EC12}
\end{align}
where ${{\sigma}}^{{p}}_{ij}$ and ${{\sigma}}^{{s}}_{ij}$ are {\em the phonon and phason stress tensors}, respectively, and $f_i^{{p}}$ is {\em the conventional (phonon) body force density} and $f_i^{{s}}$ is {\em a generalized (phason) body force density}. The comma denotes differentiation with respect to the material coordinates. We note that the phonon stress tensor is symmetric, ${{\sigma}}^{{p}}_{ij}={{\sigma}}^{{p}}_{ji}$, while the phason stress tensor is allowed to be asymmetric, ${{\sigma}}^{{s}}_{ij}\neq{{\sigma}}^{{s}}_{ji}$ (see, e.g., \cite{Ding1993}).
The phonon and phason distortion tensors, ${{\upgamma}}_{kl}$ and ${{\upkappa}}_{kl}$, are defined as the spatial gradients of ${u}_{k}^{{p}}$ and ${u}_{k}^{{s}}$, respectively
\begin{align}
\label{B-u}
{{\upgamma}}_{kl}={\up}_{k,l},\quad
{{\upkappa}}_{kl}={\us}_{k,l}.
\end{align}
The constitutive relations between the stress tensors and the distortion tensors are
\begin{align}
{{\sigma}}^{{p}}_{ij}={C}_{ijkl}{{\upgamma}}_{kl}+{D}_{ijkl}{{\upkappa}}_{kl},\quad
{{\sigma}}^{{s}}_{ij}={D}_{klij}{{\upgamma}}_{kl}+{E}_{ijkl}{{\upkappa}}_{kl}.
\label{CRe12}
\end{align}
In the above backdrop, we investigate the null Lagrangians of the following type. \\
With $\up:\Omega\to\mbb{R}^3$ and $\us:\Omega\to\mbb{R}^3$, the Lagrangian is given by
\begin{align}
{\Psi}({\nabla} \up,{\nabla} \us)={{\mathcal{E}}}({{\upgamma}},{{\upkappa}})=\frac{1}{2}{{\upgamma}}{\cdot} {\mbb{C}}[{{\upgamma}}]+\frac{1}{2}{{\upkappa}}{\cdot} {\mathbb{E}}[{{\upkappa}}]+{{\upgamma}}{\cdot} {{\mathbb{D}}}[{{\upkappa}}],
\label{quasielast}
\end{align}
\[
\text{where }
{\upgamma}={\nabla}\up,{\upkappa}={\nabla}\us,
\]
and the three constitutive tensors possess the symmetries
\begin{align}
\label{C-Sym}
{C}_{ijkl}={C}_{klij}={C}_{ijlk}={C}_{jikl},
\quad
{D}_{ijkl}={D}_{jikl},
\quad
{E}_{ijkl}={E}_{klij}.
\end{align}
Here, ${C}_{ijkl}$ is the tensor of the elastic moduli of phonons, ${E}_{ijkl}$ is the tensor of the elastic moduli of phasons, and ${D}_{ijkl}$ is the tensor of the elastic moduli of the phonon-phason coupling.
The symmetries of the tensors of the elastic constants can be simplified according to the specific type of the considered quasicrystal (see e.g. \cite{Hu2000, Fan11}). 
After we substitute 
\eqref{CRe12}
and \eqref{B-u} into 
\eqref{EC12}, we obtain
\begin{align}
\label{EOM1}
&{C}_{ijkl} {u}_{k,lj}^{{p}}+{D}_{ijkl} {u}_{k,lj}^{{s}}=-f_i^{{p}},\\
\label{EOM2}
&{D}_{klij} {u}_{k,lj}^{{p}}+ {E}_{ijkl} {u}_{k,lj}^{{s}}=-f_i^{{s}}.
\end{align}

In the context of the Euler--Lagrange equation \eqref{E:ele}, here the dependent variable is $\bsym{y}=({{\bsym{u}}},\bsym{v})$, where ${{\bsym{u}}}$ and $\bsym{v}$ both have dimension $3$, so that $N=6$ here. Rewriting 
\eqref{E:flg} for ${{\bsym{u}}}$ and $\bsym{v}$ we get,
\begin{multline}
\label{E:fnl}
{\Psi}(\bsym{x},{{\bsym{u}}},\bsym{v},\nabla{{{\bsym{u}}}},\nabla{\bsym{v}})=\frac{1}{2}\mathscr{D}_1(\alpha,\beta;i,j){u}_{i,\alpha}{u}_{j,\beta}+\frac{1}{2}\mathscr{D}_2(\alpha,\beta;i',j')v_{i,\alpha}v_{j,\beta}\\
+\mathscr{D}_3(\alpha,\beta;i,j'){u}_{i,\alpha}v_{j,\beta}+\mathscr{D}_1(\alpha;i){u}_{i,\alpha}+\mathscr{D}_2(\alpha;i')v_{i,\alpha}+\frac{1}{2}\mathscr{D}(0;0)
\end{multline}
where $1\le i,j\le3$ and $i'=i+3,\,j'=j+3$.\\
Also rewriting given Lagrangian in terms of ${u}_{i,\alpha}=\nabla{{\bsym{u}}}$ and $v_{i,\alpha}=\nabla\bsym{v}$ we get,
\begin{equation}
\label{E:fp}
{\Psi}(\bsym{x},{{\bsym{u}}},\bsym{v},\nabla{{\bsym{u}}},\nabla\bsym{v})=\frac{1}{2}{C}_{i\alpha j\beta}{u}_{i,\alpha}{u}_{j,\beta}+\frac{1}{2}{E}_{i\alpha j\beta}v_{i,\alpha}v_{j,\beta}+{D}_{i\alpha j\beta}{u}_{i,\alpha}v_{j,\beta}.
\end{equation}
By comparing \eqref{E:fp} and \eqref{E:fnl} we get,
\begin{alignat}{8}
\begin{vmatrix}
S^\alpha_i&S^\beta_i\\
S^\alpha_j&S^\beta_j
\end{vmatrix}
&=\mathscr{D}_1(\alpha,\beta;i,j)={C}_{i\alpha j\beta}, \quad&
\begin{vmatrix}
S^\alpha_{i'}&S^\beta_{i'}\\
S^\alpha_{j'}&S^\beta_{j'}
\end{vmatrix}
&=\mathscr{D}_2(\alpha,\beta;i',j')={E}_{i\alpha j\beta},\label{E:fs}\\
\begin{vmatrix}
S^\alpha_i&S^\beta_i\\
S^\alpha_{j'}&S^\beta_{j'}
\end{vmatrix}
&=\mathscr{D}_3(\alpha,\beta;i,j')={D}_{i\alpha j\beta},\quad&
\begin{vmatrix}
S^\alpha_i&S^\beta_i\\
S^\alpha_{|\beta}&S^\beta_{|\beta}
\end{vmatrix}
&=\mathscr{D}_1(\alpha;i)=0,\label{E:ffr}\\
\begin{vmatrix}
S^\alpha_{i'}&S^\beta_{i'}\\
S^\alpha_{|\beta}&S^\beta_{|\beta}
\end{vmatrix}
&=\mathscr{D}_2(\alpha;i')=0, \quad&
\begin{vmatrix}
S^\alpha_{|\alpha}&S^\beta_{|\alpha}\\
S^\alpha_{|\beta}&S^\beta_{|\beta}
\end{vmatrix}
&=\mathscr{D}(0;0)=0\label{E:fsx}.
\end{alignat}
These are necessery and sufficient conditions on $\mathbb C,{\mathbb{E}}$ and ${\mathbb{D}}$ so that \eqref{E:fp} becomes null Lagrangian within the family stated by H. Rund \cite{Rund66}. 
In the following we further explore possible explicit conditions on ${\mathbb{C}},{\mathbb{E}}$ and ${\mathbb{D}}$.

From 
\eqref{E:fs} and \eqref{E:ffr}${}_1$ we get the following restrictions on ${\mathbb{C}},{\mathbb{E}}$ and ${\mathbb{D}}$,
\begin{gather}
{C}_{ijkl}={C}_{klij},\;{E}_{ijkl}={E}_{klij},\;{D}_{ijkl}={D}_{klij}\label{E:fsv}\\
{C}_{ijkl}=-{C}_{kjil},\;{E}_{ijkl}=-{E}_{kjil},\;{D}_{ijkl}=-{D}_{kjil}\label{E:fe}\\
{C}_{ijkl}=-{C}_{ilkj},\;{E}_{ijkl}=-{E}_{ilkj},\;{D}_{ijkl}=-{D}_{ilkj}\label{E:fn}
\end{gather}
\begin{align}
and\quad {C}_{ijkl}=0,\;&{E}_{ijkl}=0&\text{if $i=k$ or $j=l$}\label{E:ftn}\\
&{D}_{ijkl}=0&\text{if $j=l$}\label{E:fev}
\end{align}
Also we recall that ${C}_{ijkl}={C}_{jikl}$ and ${D}_{ijkl}={D}_{jikl}$ are given symmetry restrictions. As a result, we get
\begin{gather*}
{C}_{ijkl}=-{C}_{kjil}=-{C}_{jkil}={C}_{ikjl},\\
\therefore\quad {C}_{ijkl}=-{C}_{kjil}=-{C}_{jkil}=-{C}_{jikl}=-{C}_{ijkl}\\
\implies {C}_{ijkl}=0\implies{\mathbb{C}}=0.\\
\text{Similarly,}\quad{\mathbb{D}}=0.
\end{gather*}

Therefore the necessary 
sufficient conditions for null Lagrangian of above type \eqref{quasielast} (with \eqref{C-Sym}) and within the family stated by H. Rund \cite{Rund66} are
\begin{gather}
{\mathbb{C}}=0,\quad
{\mathbb{D}}=0,\\
{E}_{ijkl}={E}_{klij},\quad
{E}_{ijkl}=-{E}_{kjil},\quad
{E}_{ijkl}=-{E}_{ilkj},
\end{gather}
\begin{align}
\text{and}\quad {E}_{ijkl}=0\quad\text{if $i=k$ or $j=l$},
\end{align}
and the null Lagrangian is given by
\begin{align}
{\Psi}(\bsym{x},{{\bsym{u}}},\bsym{v},\nabla{{\bsym{u}}},\nabla\bsym{v})=
\frac{1}{2}{\upkappa}\cdot{\mathbb{E}}[{\upkappa}].
\label{psiquasielas}
\end{align}
\begin{remark}
It can be easily verified by direct calculation that ${\Psi}$ in \eqref{psiquasielas} is indeed a null Lagrangian.
Again, in contrast to the case of linear elasticity above null Lagrangian can be interpreted as stored energy functionals in the linear generalized elasticity theory of quasicrystals.
\end{remark}

\section{Linear electro-magneto-elasticity}
With ${\bsym{u}}:\Omega\to\mbb{R}^3$, ${\varphi}:\Omega\to\mbb{R}$ and ${\psi}:\Omega\to\mbb{R}$, 
assuming that ${{C}}_{ijlk}, {{P}}_{ikl}, {{Q}}_{ikl}, {E}_{il}, {{B}}_{il}, {{A}}_{il}$ are the elastic, piezoelectric, piezomagnetic, dielectric, magnetic permeability and electromagnetic coupling constants, respectively,
the Lagrangian (enthalpy functional) for a linear electro-magneto-elastic material is given by \cite{BCJ1964}
\begin{align}
2H({\bsym{u}},{\bsym{e}},{\bsym{h}})={{C}}_{ijkl}{u}_{i,j}{u}_{k,l}-{E}_{ij}{\varphi}_{,i}{\varphi}_{,j}-{{B}}_{ij}{\psi}_{,i}{\psi}_{,j}\notag\\
+2{{P}}_{kij}{\varphi}_{,k}{u}_{i,j}
+2{{Q}}_{kij}{\varphi}_{,k}{u}_{i,j}
-2{{A}}_{ij}{\varphi}_{,i}{\psi}_{,j}.
\label{E:sp}
\end{align}
where the material constants satisfy the following conditions of symmetry
\[
{{C}}_{ijkl}={{C}}_{jikl}={{C}}_{klij}, 
\quad
{{P}}_{kij}={{P}}_{kji}, 
\quad
{{Q}}_{kij}={{Q}}_{kji}, 
\quad
{E}_{ij}={E}_{ji},
\quad
{{A}}_{ij}={{A}}_{ij},
\quad
{{B}}_{ij}={{B}}_{ji}.
\]
Here, $\delta_{jk}$ is the Kronecker symbol and
summation over repeated indices is implied.
In other words, with ${\bsym{e}}=-{\nabla}{\varphi}$ and ${\bsym{h}}=-{\nabla}{\psi}$,
\begin{align}
H({\bsym{u}},{\bsym{e}},{\bsym{h}})=\frac{1}{2}{{\mathbb{C}}}{\upepsilon}:{\upepsilon}
-\frac{1}{2}{{\mbf{E}}}{\bsym{e}}{\cdot}{\bsym{e}}
-\frac{1}{2}{\mbf{B}}{\bsym{h}}{\cdot}{\bsym{h}}
-{{\mathtt{P}}}{\upepsilon}{\cdot}{\bsym{e}}
-{{\mathtt{Q}}}{\upepsilon}{\cdot}{\bsym{h}}
-{\mbf{A}}{\bsym{e}}{\cdot}{\bsym{h}}.
\label{E:spalt}
\end{align}
In the theory of electro-magneto-elasticity, 
the stress \cite{BCJ1964}
is
\begin{align}
{\bsym{\sigma}}&={{\mathbb{C}}}\nabla{\bsym{u}}+{{\mathtt{P}}}^{\top}\nabla{\varphi}+{{\mathtt{Q}}}^{\top}\nabla{\psi}
={{\mathbb{C}}}{\upepsilon}-{{\mathtt{P}}}^{\top}{\bsym{e}}-{{\mathtt{Q}}}^{\top}{\bsym{h}},\\
{\bsym{d}}&={{\mathtt{P}}}\nabla{\bsym{u}}-{{\mbf{E}}}\nabla{\varphi}-{\mbf{A}}\nabla{\psi}
={{\mathtt{P}}}{\upepsilon}+{{\mbf{E}}}{\bsym{e}}+{\mbf{A}}{\bsym{h}},\\
\bsym{b}&={{\mathtt{Q}}}\nabla{\bsym{u}}-{\mbf{A}}\nabla{\varphi}-{\mbf{B}}\nabla{\psi}
={{\mathtt{Q}}}{\upepsilon}+{\mbf{A}}{\bsym{e}}+{\mbf{B}}{\bsym{h}}.
\label{stressemmat}
\end{align}
The classical variation of electromagnetic enthalpy functional
\[
\mcal{H}({\bsym{u}},{\varphi},{\psi}) =\int_\Omega H({\bsym{u}},{\bsym{e}},{\bsym{h}})dv,
\]
provides the Euler--Lagrange differential equations.
In the context of the Euler--Lagrange equation \eqref{E:ele}, the dependent variable is $\bsym{y}=({{\bsym{u}}},{\varphi},{\psi})$, while ${{\bsym{u}}}$ has dimension $3$, ${\varphi}$ and ${\psi}$ have dimension $1$, so that $N=5$. Rewriting 
\eqref{E:flg} for ${{\bsym{u}}},{\varphi}$ and ${\psi}$ we get,
\begin{multline}
\label{E:snl}
{\Psi}(
\nabla{{{\bsym{u}}}},\nabla{{\varphi}},\nabla{{\psi}})=\frac{1}{2}\mathscr{D}_1(\alpha,\beta;i,j){u}_{i,\alpha}{u}_{j,\beta}+\frac{1}{2}\mathscr{D}_2(\alpha,\beta;4,4){\varphi}_{,\alpha}{\varphi}_{,\beta}+\frac{1}{2}\mathscr{D}_3(\alpha,\beta;5,5){\psi}_{,\alpha}{\psi}_{,\beta}\\
+\mathscr{D}_4(\alpha,\beta;4,i){u}_{i,\beta}{\varphi}_{,\alpha}+\mathscr{D}_5(\alpha,\beta;5,j){u}_{i,\beta}{\psi}_{,\alpha}+\mathscr{D}_6(\alpha,\beta;4,5){\varphi}_{,\alpha}{\psi}_{,\beta}\\
+\mathscr{D}_1(\alpha;i){u}_{i,\alpha}+\mathscr{D}_2(\alpha;4){\varphi}_{,\alpha}+\mathscr{D}_3(\alpha;5){\psi}_{,\alpha}+\frac{1}{2}\mathscr{D}(0;0)
\end{multline}
where $1\le i,j\le3$.\\
Comparing \eqref{E:sp} and \eqref{E:snl} we get,
\begin{alignat}{8}
\begin{vmatrix}
S^\alpha_i&S^\beta_i\\
S^\alpha_j&S^\beta_j
\end{vmatrix}
&=\mathscr{D}_1(\alpha,\beta;i,j)={{C}}_{i\alpha j\beta},\quad&
\begin{vmatrix}
S^\alpha_4&S^\beta_4\\
S^\alpha_4&S^\beta_4
\end{vmatrix}
&=\mathscr{D}_2(\alpha,\beta;4,4)=-{E}_{\alpha\beta}\label{E:ss}\\
\begin{vmatrix}
S^\alpha_5&S^\beta_5\\
S^\alpha_5&S^\beta_5
\end{vmatrix}
&=\mathscr{D}_3(\alpha,\beta;5,5)=-{B}_{\alpha\beta},\quad&
\begin{vmatrix}
S^\alpha_4&S^\beta_4\\
S^\alpha_i&S^\beta_i
\end{vmatrix}
&=\mathscr{D}_4(\alpha,\beta;4,i)={{P}}_{\alpha i\beta}\label{E:sfr}\\
\begin{vmatrix}
S^\alpha_5&S^\beta_5\\
S^\alpha_i&S^\beta_i
\end{vmatrix}
&=\mathscr{D}_5(\alpha,\beta;5,i)={{Q}}_{\alpha i\beta},\quad&
\begin{vmatrix}
S^\alpha_4&S^\beta_4\\
S^\alpha_5&S^\beta_5
\end{vmatrix}
&=\mathscr{D}_6(\alpha,\beta;4,5)=-{A}_{\alpha\beta}\label{E:ssx}\\
\begin{vmatrix}
S^\alpha_i&S^\beta_i\\
S^\alpha_{|\beta}&S^\beta_{|\beta}
\end{vmatrix}
&=\mathscr{D}_1(\alpha;i)=0,\quad&
\begin{vmatrix}
S^\alpha_4&S^\beta_4\\
S^\alpha_{|\beta}&S^\beta_{|\beta}
\end{vmatrix}
&=\mathscr{D}_2(\alpha;4)=0,\\
\begin{vmatrix}
S^\alpha_5&S^\beta_5\\
S^\alpha_{|\beta}&S^\beta_{|\beta}
\end{vmatrix}
&=\mathscr{D}_3(\alpha;5)=0,\quad&
\begin{vmatrix}
S^\alpha_{|\alpha}&S^\beta_{|\alpha}\\
S^\alpha_{|\beta}&S^\beta_{|\beta}
\end{vmatrix}
&=\mathscr{D}(0;0)=0
\end{alignat}
Thus, within the family of null Lagrangians stated by H. Rund \cite{Rund66}, we find that these are necessary and sufficient conditions on $\mathbb C,{\mbf{E}},{\mbf{B}}, {{\mathtt{P}}},{{\mathtt{Q}}},$ and ${\mbf{A}}$ so that \eqref{E:sp} becomes a null Lagrangian. As in the case of previous section, we obtain the explicit form of conditions on the material constants $\mathbb C,\mbf{E},{\mbf{B}},{{\mathtt{P}}},{{\mathtt{Q}}}$ and ${\mbf{A}}$.
From \eqref{E:ss}${}_1$
we get 
$\mathbb C=0$.
From \eqref{E:ss}${}_2$ and \eqref{E:sfr}${}_1$ we get ${\mbf{E}}=0$ and ${\mbf{B}}=0$ respectively.
From \eqref{E:sfr}${}_2$ and \eqref{E:ssx}${}_1$ we get the following restrictions on ${{\mathtt{P}}}$ and ${{\mathtt{Q}}}$,
\begin{align}
{{P}}_{kij}=-{{P}}_{jik},\;&{{Q}}_{kij}=-{{Q}}_{jik}\label{E:sev}\\
{{P}}_{kij}=0,\;&{{Q}}_{kij}=0&\text{if $j=k$}\label{E:stv}
\end{align}
Also given ${{P}}_{kij}={{P}}_{kji},{{Q}}_{kij}={{Q}}_{kji}$.
\begin{gather*}
\therefore\quad {{P}}_{kij}=-{{P}}_{jik}=-{{P}}_{jki}={{P}}_{ikj}={{P}}_{ijk}=-{{P}}_{kji}=-{{P}}_{kij}\\
\implies{{\mathtt{P}}}=0.
\end{gather*}
Same holds for ${{\mathtt{Q}}}$ so ${{\mathtt{Q}}}=0$ as well.
From \eqref{E:ssx}${}_2$ we get the following restrictons on ${\mbf{A}}$,
\begin{align*}
{A}_{ij}=-{A}_{ji},\\
\text{ so that }{A}_{ij}=0\text{ if $i=j$}.
\end{align*}
But it is given that ${A}_{ij}={A}_{ji}$, $\implies\mbf{A}=0$.
Therefore,
within the family of null Lagrangians stated by H. Rund \cite{Rund66}, we do not obtain any non-zero null Lagrangians in this case.
This conclusion resonates with that in linear elasticity where null Lagrangians cannot be interpreted as stored energy functionals. 

\section{Conclusion}
In this paper, we explored the concept of null Lagrangians within the linear framework of generalized elasticity.
In particular, we have constructed a family of null Lagrangians for linear Cosserat or micropolar elastic media as well as linear quasicrystal models. We found out that there does not exist a non-trivial null Lagrangian for (linear) piezoelectric and piezomagnetic media, or in general linear electro-magneto-elasticity. 
For the cases where the non-trivial null Lagrangian exists, we provide a split of the stored energy where the null Lagrangian part of it vanishes if and only if the relevant elasticity tensor obeys certain symmetry conditions as a reminder of Cauchy relations in classical elasticity. We have presented analysis of isotropic models and few other special cases as well.
At the end of the analyses presented, a question arises about the extension to nonlinear models of the discussed formulations; this is currently under investigation \cite{Basak2}.

\section*{Acknowledgments}
BLS acknowledges the partial support of SERB MATRICS grant MTR/2017/000013.

\appendix

\section{Characterization of null Lagrangians}
\label{S:cnl}

If the right-hand side of \eqref{E:ele} is to vanish for all values of $\frac{{\pd}^2y^j}{{\pd} x^{{{\beta}}}{\pd} x^{{{\gamma}}}}$, the coefficient of these quantities which is symmetric in ${{{\beta}}}$ and ${{{\gamma}}}$, must be skew-symmetric in these indices, i.e.,
\[
\frac{{\pd}^2{\Psi}}{{\pd} y^j_{{{\beta}}}{\pd} y^k_{{{\gamma}}}}=-\frac{{\pd}^2{\Psi}}{{\pd} y^j_{{{\gamma}}}{\pd} y^k_{{{\beta}}}}.
\]
Following the analysis of \cite{Rund66},
in general, ${\Psi}$ is a polynomial in its dependence on the derivative of $y$. 
In the class of Lagrangians which are interesting in this document, it is sufficient to consider $M=2.$
Thus, ${\Psi}$ can be written as 
\begin{equation}
\label{E:psid}
{\Psi}={\Psi}^{(2)}+{\Psi}^{(1)}+\Phi,
\end{equation}
\begin{equation}
\label{E:psi}
\text{where }
{\Psi}^{(1)}:=A^{\alpha_1}_{i_1}y^{i_1}_{\alpha_1},\quad
{\Psi}^{(2)}:=\frac{1}{2}A^{\alpha_1\alpha_2}_{i_1i_2}y^{i_1}_{\alpha_1}y^{i_2}_{\alpha_2}.
\end{equation}
Here, $\Phi,\mbf{A},{{\mathbb{A}}}$ are functions of $\bsym{x}$ and $\bsym{y}$ of class $\mathcal{C}^2$. Clearly
$A^{\alpha_1\alpha_2}_{i_1i_2}=A^{\alpha_2\alpha_1}_{i_2i_1}.$
Then \eqref{E:eqv} becomes
\begin{equation}
\label{E:eql}
\sum_{n=1}^2{\mathscr{E}}_k({\Psi}^{(n)})+{\mathscr{E}}_k(\Phi)\equiv0.
\end{equation}
On substituting from \eqref{E:ele} we see that this represents a set of PDEs for $\mbf{A}$ and ${{\mathbb{A}}}$, whose precise form we have to determine.

From \eqref{E:psi} we have
$\frac{{\pd}{\Psi}^{(2)}}{{\pd} y^k_{{{\gamma}}}}=A^{{{{\gamma}}}\alpha_2}_{ki_2}y^{i_2}_{\alpha_2}\quad\frac{{\pd}{\Psi}^{(1)}}{{\pd} y^k_{{{\gamma}}}}=A^{{{\gamma}}}_k.$
Therefore
\begin{align*}
{\mathscr{E}}_k({\Psi}^{(2)})&=\frac{\pd}{{\pd} x^{{{\gamma}}}}\left(A^{{{{\gamma}}}\alpha_2}_{ki_2}\right)y^{i_2}_{\alpha_2}+\left[\frac{\pd}{{\pd} y^{i_1}}\left(A^{\alpha_1\alpha_2}_{ki_2}\right)-\frac{1}{2}\frac{\pd}{{\pd} y^k}\left(A^{\alpha_1\alpha_2}_{i_1i_2}\right)\right]y^{i_1}_{\alpha_1}y^{i_2}_{\alpha_2}+A^{{{{\gamma}}}{{{\beta}}}}_{kj}\frac{{\pd}^2y^j}{{\pd} x^{{{\beta}}}{\pd} x^{{{\gamma}}}},\\
{\mathscr{E}}_k({\Psi}^{(1)})&=\frac{\pd}{{\pd} x^{{{\gamma}}}}\left(A^{{{\gamma}}}_k\right)+\left[\frac{\pd}{{\pd} y^{i_1}}\left(A^{\alpha_1}_k\right)-\frac{\pd}{{\pd} y^k}\left(A^{\alpha_1}_{i_1}\right)\right]y^{i_1}_{\alpha_1}.
\end{align*}
From \eqref{E:eql} we get,
\begin{align}
&\left[\frac{\pd}{{\pd} x^{{{\gamma}}}}\left(A^{{{{\gamma}}}\alpha_2}_{ki_2}\right)+\frac{\pd}{{\pd} y^{i_2}}\left(A^{\alpha_2}_k\right)-\frac{\pd}{{\pd} y^k}\left(A^{\alpha_2}_{i_2}\right)\right]y^{i_2}_{\alpha_2}+\left[\frac{\pd}{{\pd} y^{i_1}}\left(A^{\alpha_1\alpha_2}_{ki_2}\right)-\frac{1}{2}\frac{\pd}{{\pd} y^k}\left(A^{\alpha_1\alpha_2}_{i_1i_2}\right)\right]y^{i_1}_{\alpha_1}y^{i_2}_{\alpha_2}\notag\\
&\quad\quad\quad\quad\quad+A^{{{{\gamma}}}{{{\beta}}}}_{kj}\frac{{\pd}^2y^j}{{\pd} x^{{{\beta}}}{\pd} x^{{{\gamma}}}}+\left[\frac{\pd}{{\pd} x^{{{\gamma}}}}\left(A^{{{\gamma}}}_k\right)-\frac{{\pd}\Phi}{{\pd} y^k}\right]=0.\label{E:eleq}
\end{align}
This holds for all values of $y^i_\alpha$ and $\frac{{\pd}^2y^j}{{\pd} x^{{{\beta}}}{\pd} x^{{{\gamma}}}}$. So every term is zero separately. In particular,
\begin{align}
\left[\frac{\pd}{{\pd} x^{{{\gamma}}}}\left(A^{{{{\gamma}}}\alpha_2}_{ki_2}\right)+\frac{\pd}{{\pd} y^{i_2}}\left(A^{\alpha_2}_k\right)-\frac{\pd}{{\pd} y^k}\left(A^{\alpha_2}_{i_2}\right)\right]=0,\quad
\left[\frac{\pd}{{\pd} x^{{{\gamma}}}}\left(A^{{{\gamma}}}_k\right)-\frac{{\pd}\Phi}{{\pd} y^k}\right]=0.
\label{E:rl12}
\end{align}
for $k=1,\dots,N.$

Let us consider $3$ functions (of $\mathcal{C}^2$ smoothness) $\{S^\alpha(\bsym{x},\bsym{y})\}_{\alpha=1,2,3}$ and define 
\[
S^\alpha_i:=\frac{{\pd} S^\alpha}{{\pd} y^i}; \quad 
S^\alpha_{|\beta}:=\frac{{\pd} S^\alpha}{{\pd} x^\beta}.
\]
Recall
$\mfk{D}$s as defined by \eqref{Daaii}. 
Therefore,
\begin{align}
\frac{{\pd} \mfk{D}(\alpha_1;i_1)}{{\pd} x^{\alpha_1}}&=
\begin{vmatrix}
S^{\alpha_1}_{i_1|\alpha_1}&S^{\alpha_2}_{i_1|\alpha_1}\\
S^{\alpha_1}_{|\alpha_2}&S^{\alpha_2}_{|\alpha_2}
\end{vmatrix}
+
\begin{vmatrix}
S^{\alpha_1}_{i_1}&S^{\alpha_2}_{i_1}\\
S^{\alpha_1}_{|\alpha_1\alpha_2}&S^{\alpha_2}_{|\alpha_1\alpha_2}
\end{vmatrix}\notag\\
&\text{($\alpha_1,\alpha_2$ both being dummy indices, the latter determinant is $0$)}\\
&=\frac{\pd}{{\pd} y^{i_1}}
\begin{vmatrix}
S^{\alpha_1}_{|\alpha_1}&S^{\alpha_2}_{|\alpha_1}\\
S^{\alpha_1}_{|\alpha_2}&S^{\alpha_2}_{|\alpha_2}
\end{vmatrix}
-\begin{vmatrix}
S^{\alpha_1}_{|\alpha_1}&S^{\alpha_2}_{|\alpha_1}\\
S^{\alpha_1}_{i_1|\alpha_2}&S^{\alpha_2}_{i_1|\alpha_2}
\end{vmatrix}=\frac{{\pd} \mfk{D}(0;0)}{{\pd} y^{i_1}}-\frac{{\pd} \mfk{D}(\alpha_1;i_1)}{{\pd} x^{\alpha_1}}\notag\\
\implies2\frac{{\pd} \mfk{D}(\alpha_1;i_1)}{{\pd} x^{\alpha_1}}&=\frac{{\pd} \mfk{D}(0;0)}{{\pd} y^{i_1}}.
\label{E:rq}
\end{align}

Also,
\begin{align*}
\frac{{\pd} \mfk{D}(\alpha_1,\alpha_2;i_1,i_2)}{{\pd} x^{\alpha_2}}&=
\begin{vmatrix}
S^{\alpha_1}_{i_1|\alpha_2}&S^{\alpha_2}_{i_1|\alpha_2}\\
S^{\alpha_1}_{i_2}&S^{\alpha_2}_{i_2}
\end{vmatrix}
+
\begin{vmatrix}
S^{\alpha_1}_{i_1}&S^{\alpha_2}_{i_1}\\
S^{\alpha_1}_{i_2|\alpha_2}&S^{\alpha_2}_{i_2|\alpha_2}
\end{vmatrix}\\
&=\frac{\pd}{{\pd} y^{i_1}}
\begin{vmatrix}
S^{\alpha_1}_{|\alpha_2}&S^{\alpha_2}_{|\alpha_2}\\
S^{\alpha_1}_{i_2}&S^{\alpha_2}_{i_2}
\end{vmatrix}
+\frac{\pd}{{\pd} y^{i_2}}
\begin{vmatrix}
S^{\alpha_1}_{i_1}&S^{\alpha_2}_{i_1}\\
S^{\alpha_1}_{|\alpha_2}&S^{\alpha_2}_{|\alpha_2}
\end{vmatrix}
-
\begin{vmatrix}
S^{\alpha_1}_{|\alpha_2}&S^{\alpha_2}_{|\alpha_2}\\
S^{\alpha_1}_{i_1i_2}&S^{\alpha_2}_{i_1i_2}
\end{vmatrix}
-
\begin{vmatrix}
S^{\alpha_1}_{i_1i_2}&S^{\alpha_2}_{i_1i_2}\\
S^{\alpha_1}_{|\alpha_2}&S^{\alpha_2}_{|\alpha_2}
\end{vmatrix}\\
&=-\frac{{\pd} \mfk{D}(\alpha_1;i_2)}{{\pd} y^{i_1}}+\frac{{\pd} \mfk{D}(\alpha_1;i_1)}{{\pd} y^{i_2}}\\
&\text{(the remaining two determinants cancel each other)}.
\end{align*}
Comparing this with \eqref{E:rl12}${}_1$ we get
\begin{equation}
\label{E:ad}
A^{\alpha_1\alpha_2}_{i_1i_2}=\mfk{D}(\alpha_1,\alpha_2;i_1,i_2),\quad
A^{\alpha_1}_{i_1}=\mfk{D}(\alpha_1;i_1).
\end{equation}
From \eqref{E:rq},
$\frac{1}{2}\frac{{\pd} \mfk{D}(0;0)}{{\pd} y^{i_1}}=\frac{{\pd} A^{\alpha_1}_{i_1}}{{\pd} x^{\alpha_1}}=\frac{{\pd}\Phi}{{\pd} y^{i_1}}$
by \eqref{E:rl12}${}_2$,
which upon integrating yields
\[
\Phi=\frac{1}{2}\mfk{D}(0;0).
\]
Therefore, reverting back to \eqref{E:psid} we get \eqref{E:flg}.

Using \eqref{E:psid}, \eqref{E:psi}, \eqref{E:eleq} and \eqref{E:ad} we can summarize the result as:
\begin{theorem}[\cite{Rund66}, pg 257-258]
Any function of the form \eqref{E:flg} satisfies Euler-Lagrange equation \eqref{E:eqv} identically, where $\mathcal{C}^2$ functions $S^\alpha(\bsym{x},\bsym{y})$ are entirely arbitrary.
\label{mainthm}
\end{theorem}

\end{document}